%% file: main.tex
\documentclass[11pt, letter]{article}
\usepackage[T1]{fontenc}
\usepackage[utf8]{inputenc}
\pdfoutput=1

\def\withcolors{1}
\def\withnotes{1}

\ifnum\withcolors=0
  \newcommand{\newest}[1]{{\color{orange} {#1}}} 
\else

  \newcommand{\newest}[1]{{{#1}}}
\fi

\usepackage[colorinlistoftodos,textsize=scriptsize]{todonotes}
\ifnum\withnotes=1

\else

\fi

\input{glodef}

\input{locdef}

\title{Measuring Quantum Entropy}
\author{Jayadev Acharya\footnote{The authors are listed in alphabetical order. Part of the work performed when II was a student at Cornell University. JA was supported by a start-up grant from Cornell University. II, NVS, ABW were supported by US National Science Foundation under grants CCF-1704443 and CCF-1513858.}\\
Cornell University\\
\tt{acharya@cornell.edu}
\and
Ibrahim Issa$^*$\\
EPFL\\
\tt{ibrahim.issa@epfl.ch}
\and
Nirmal V. Shende$^*$\\
Cornell University\\
\tt{nvs25@cornell.edu}
\and 
Aaron B. Wagner$^*$\\
Cornell University\\
\tt{wagner@cornell.edu}
}

\begin{document}
\addtocounter{page}{-1}
\maketitle 
\thispagestyle{empty}
\input{abstract.tex}

\newpage

\newpage
\renewcommand*{\arraystretch}{1.3}

\input{intro}

\input{our-techniques}
\input{related-work}

\input{formulation}

\input{preliminaries}

\input{polynomials}

\input{wss}
\input{upper-bound-techniques.tex}

\input{integral-alpha}

\input{achievability-integral}

\input{converse-integral}

\input{von-neumann}

\input{non-integral-large-alpha}

\input{non-integral-small-alpha}

\input{converse-large-alpha-empirical}

\input{converse-small-alpha-empirical}

\input{acknowledgements.tex}

\bibliographystyle{alpha}

\bibliography{abr,masterref}

\appendix

\input{proof-lemma-lengthpartition}

\input{app-majorization}

\input{app-falling.tex}

\input{app-bias-von-neumann.tex}

\input{app-lipschitz-entropy.tex}

\input{app-lipschitz-tv.tex}

\input{app-large-alpha.tex}

\input{app-small-alpha.tex}

\input{app-lemma-conditioning}

\input{app-small-alpha-converse.tex}

\end{document}

%% file: glodef.tex
\usepackage{lmodern}
\usepackage{xspace}
\usepackage[protrusion=true,expansion=true]{microtype}
\usepackage{fullpage}
\usepackage{amsfonts,amsmath,amssymb,amsthm, pbox}

\usepackage[colorlinks,citecolor=blue,bookmarks=true]{hyperref}

\usepackage{multirow}

\usepackage{dsfont} 
\usepackage{fullpage}

\usepackage{algorithmicx,algpseudocode,  algorithm}

\usepackage[shortlabels]{enumitem}
\setitemize{noitemsep,topsep=0pt,parsep=0pt,partopsep=0pt}
\setenumerate{noitemsep,topsep=0pt,parsep=0pt,partopsep=0pt}

\makeatletter
\@ifundefined{theorem}{
  \theoremstyle{definition}
  \newtheorem{definition}{Definition}
  \theoremstyle{plain}
  \newtheorem{theorem}{Theorem}
  \newtheorem{corollary}{Corollary}
  \newtheorem{lemma}{Lemma}

  \theoremstyle{remark}
  \newtheorem{remark}{Remark}

}{}
\makeatother

\newcommand{\ignore}[1]{}

\newcommand{\EE}{\mathbb{E}}
\newcommand{\CC}{\mathbb{C}}
\newcommand{\NN}{\mathbb{N}}

\newcommand{\RR}{\mathbb{R}}

\newcommand{\expectation}[1]{\EE\left[#1\right]}
\newcommand{\variance}[1]{\mathrm{Var}\left(#1\right)}

\def \cE     {{\cal E}}

\def \cL     {{\cal L}}

\def \cS     {{\cal S}}

\def \cX     {{\cal X}}

\newcommand{\Var}{{\rm Var}}

\def \eqed    {\eqno{\qed}}

\def \upto  {{,}\ldots{,}}

\newcommand{\absv}[1]{\left|#1\right|}

\def \Paren#1{{\left({#1}\right)}}

\def \Brack#1{{\left[{#1}\right]}}

\newcommand{\ed}{\stackrel{\text{def}}{=}}

\newcommand{\probof}[1]{\Pr\Paren{#1}}

\def\ignore#1{}

\newcommand{\bi}{\begin{itemize}}
\newcommand{\ei}{\end{itemize}}

\newcommand{\flnpwr}[2]{{#1^{\underline{#2}}}} 

\def\orpro{\mathop{\mathchoice
   {\vee\kern-.49em\raise.7ex\hbox{$\cdot$}\kern.4em}
   {\vee\kern-.45em\raise.63ex\hbox{$\cdot$}\kern.2em}
   {\vee\kern-.4em\raise.3ex\hbox{$\cdot$}\kern.1em}
   {\vee\kern-.35em\raise2.2ex\hbox{$\cdot$}\kern.1em}}\limits}

\def\andpro{\mathop{\mathchoice
 {\wedge\kern-.46em\lower.69ex\hbox{$\cdot$}\kern.3em}
 {\wedge\kern-.46em\lower.58ex\hbox{$\cdot$}\kern.25em}
 {\wedge\kern-.38em\lower.5ex\hbox{$\cdot$}\kern.1em}
 {\wedge\kern-.3em\lower.5ex\hbox{$\cdot$}\kern.1em}}\limits}

\def\simge{\mathrel{%
   \rlap{\raise 0.511ex \hbox{$>$}}{\lower 0.511ex \hbox{$\sim$}}}}

\def\simle{\mathrel{
   \rlap{\raise 0.511ex \hbox{$<$}}{\lower 0.511ex \hbox{$\sim$}}}}



%% file: locdef.tex
\newcommand{\eig}{\eta}
\newcommand{\eigi}[1]{\eig_{#1}}
\newcommand{\eigv}{{\boldsymbol{\eig}}}

\newcommand{\eignu}{\nu}

\newcommand{\eignuv}{{\boldsymbol{\eignu}}}
\newcommand{\eiglb}{{\boldsymbol{\eignu}}}
\newcommand{\eiglbi}[1]{\eiglb_{#1}}

\newcommand{\ent}[1]{S\Paren{#1}}
\newcommand{\entmst}{\ent{\mst}}
\newcommand{\empent}{\widehat{\ent{\mst}}}
\newcommand{\empents}[1]{\widehat{\ent{#1}}}

\newcommand{\p}{p}
\newcommand{\q}{q}

\newcommand{\rentp}[1]{H_{\renprm}(p)}
\newcommand{\dtv}[2]{d_{TV}\Paren{#1,#2}}

\newcommand{\renprm}{\alpha}
\newcommand{\rent}[2]{\cS_{#1}\Paren{#2}}

\newcommand{\rentprmmst}{\rent{\renprm}{\mst}}
\newcommand{\phash}[1]{p_{#1}^\#}

\newcommand{\Map}[2]{M_{#1}\left(#2\right)}
\newcommand{\Mapeigv}[1]{\Map {#1} {\eigv}}
\newcommand{\Malamb}[1]{\Map {#1} {\lamb}}

\newcommand{\smpcmpa}[1]{S(H_{\renprm}, \dims, \eps)}
\newcommand{\smpcmph}{S(H, \dims, \eps)}

\newcommand{\copycmpa}[1]{C(S_{\renprm}, \dims, \eps)}
\newcommand{\copycmph}{C(S, \dims, \eps)}

\newcommand{\dims}{d}
\newcommand{\eps}{\varepsilon}

\newcommand{\mst}{\rho}

\newcommand{\ns}{n}

\newcommand{\schurlx}[2]{s_{#1}\Paren{#2}}

\newcommand{\chisq}[2]{\chi^2\Paren{#1, #2}}

\newcommand{\Xon}{X^{\ns}}

\newcommand{\yd}{y_1^{\dims}}

\newcommand{\ydp}{y^{\dims}_+}

\newcommand{\mltu}{\mu}

\newcommand{\proh}[2]{p_{#1}^{\#}(#2)}

\newcommand{\eighati}[1]{\widehat{\eigv_{#1}}}

\newcommand{\lamb}{{\boldsymbol{\lambda}}}
\newcommand{\lambi}[1]{{\lamb_{#1}}}

\newcommand{\eigvi}[1]{\eigv_{#1}}

\newcommand{\swdist}[1]{SW_{#1}}

\newcommand{\tr}[1]{{\rm tr}\Paren{#1}}
\newcommand{\diffmean}{B}

\newcommand{\prp}{f}
\newcommand{\fhat}{\hat{f}}

\newcommand{\unitary}{U(d)}

%% file: abstract.tex
\begin{abstract}
The entropy of a quantum system is a measure of its randomness, and has applications in measuring quantum entanglement. We study the problem of measuring the von Neumann entropy, $S(\rho)$, and R\'enyi entropy, $S_\alpha(\rho)$ of an unknown mixed quantum state $\rho$ in $d$ dimensions, given access to independent copies of $\rho$. 

We provide an algorithm with copy complexity $O(d^{2/\alpha})$ for estimating $S_\renprm(\rho)$ for $\alpha<1$, and copy complexity $O(d^{2})$ for estimating $S(\rho)$, and  $S_\alpha(\rho)$ for non-integral $\alpha>1$. These bounds are at least quadratic in $d$, which is the order dependence on the number of copies required for learning the entire state $\rho$. For integral $\alpha>1$, on the other hand, we provide an algorithm for estimating $S_\alpha(\rho)$ with a sub-quadratic copy complexity of $O(d^{2-2/\alpha})$. We characterize the copy complexity for integral $\alpha>1$ up to constant factors by providing matching lower bounds. For other values of $\alpha$, and the von Neumann entropy, we show lower bounds on the algorithm that achieves the upper bound. This shows that we either need new algorithms for better upper bounds, or better lower bounds to tighten the results. 

For non-integral $\alpha$, and the von Neumann entropy, we consider the well known Empirical Young Diagram (EYD) algorithm, which is the analogue of empirical plug-in estimator in classical distribution estimation. As a corollary, we strengthen a lower bound on the copy complexity of the EYD algorithm for learning the maximally mixed state by showing that the lower bound holds with exponential probability (which was previously known to hold with a constant probability). For integral $\alpha>1$, we provide new concentration results of certain polynomials that arise in Kerov algebra of Young diagrams.
\end{abstract}

%% file: intro.tex
\section{Introduction}

We consider how to estimate the mixedness or noisiness
of a quantum state using measurements of independent copies of the
state. Mixed quantum states can arise in practice in various
ways. Classical stochasticity can be intentionally introduced
when the state is originally prepared. Pure states can become
mixed by a quantum measurement. And the states of the subsystems
of bipartite states can be mixed even when the overall bipartite
state is pure, which forms the basis for purification.

In the third case, the level of mixedness of the subsystems
indicates the level of entanglement in the pure, bipartite
system. The possibility of entanglement of two separated systems
is arguably the most curious, and the most powerful, way in
which quantum systems differ from classical ones. Indeed,
entanglement has been fruitfully exploited as a resource in
a number of quantum information processing 
protocols~\cite{Bennett92,Bennett93,Bennett02,Devetak04,Hsieh10}. 
The subsystems of a pure bipartite state are pure if and only if the 
bipartite state itself is unentangled, and likewise they are 
maximally mixed if and only if the bipartite state is maximally entangled. 
Thus the mixedness of the subsystems' states can be used
as a measure of entanglement of the bipartite system.

Mixedness can be measured in multiple ways. We shall use the
von Neumann and (the family of) R\'enyi entropies, which correspond to 
the classical Shannon and (the family of) R\'enyi entropies of the 
eigenvalues of the density operator of the state, respectively. 
A density matrix (or operator) $\mst$ is a complex positive semidefinite
matrix with unit trace; thus its eigenvalues are nonnegative and
sum to one.  The von Neumann entropy of a density matrix $\mst$ is 
\begin{align}
\entmst \ed -\tr{\mst\log\mst}.\nonumber	
\end{align}
For $\renprm>0, \renprm\ne1$, the R\'enyi entropy of order $\renprm$ of $\mst$ is 
\begin{align}
S_\renprm (\mst) \ed \frac1{1-\renprm}\log \tr{\mst^{\renprm}}.\nonumber	
\end{align}
In the limit of  $\renprm\to1$,
\[
\lim_{\renprm\to1} S_\renprm (\mst) = \entmst. 
\]
The classical
Shannon and R\'enyi entropies are well-accepted measures of
randomness, and can be derived 
axiomatically~\cite[pp.~25-27]{CK:IT}.
Both the classical and quantum versions can be justified 
operationally as a measure of
compressibility~\cite{CK:IT, Schumacher95, JozsaSchumacher94, Lo95}. 
The quantum versions have been explicitly proposed for 
quantifying entanglement~\cite{CardyL12}.

In principle, both the von Neumann and R\'enyi entropies
for a quantum state $\mst$ can be computed if the state
is known. We consider how to estimate these quantities 
for an unknown state 
given independent copies of the state, to which arbitrary
quantum measurements followed by arbitrary
classical computation can be applied. This problem arises when characterizing
a completely unknown system and when one seeks to experimentally
verify that a system is behaving as desired. Since generating
independent copies of a state can be quite costly in the quantum
setting~\cite{Haeffner05,Ma12}, it is desirable to minimize the number of 
independent copies of the state that are required to estimate 
the von Neumann and R\'enyi entropies to a desired precision and 
confidence. We thus adopt this \emph{copy complexity} as our 
figure-of-merit.

Using standard results in quantum state estimation, we reduce
our problem to one that is fully classical. We first describe
this fully-classical problem, which is potentially of interest
in its own right.
\subsection{Quantum-Free Formulation}
\label{sec:qff}
Let $\p$ be a distribution over $[\dims]\ed\{1,\upto\dims\}$. A property $f(\p)$ is a mapping of distributions to real numbers. A property $f$ is said to be symmetric (or label-invariant) if it is a function of only the multiset of probability values, and not the ordering. For example, the Shannon entropy $H(\p) = -\sum_i p(i)\log \p(i)$ is symmetric, since it is only a function of the probability values. 
\paragraph{Classical symmetric property estimation.} We are given independent samples $X^{\ns} \ed X_1, \ldots, X_\ns$ from an unknown distribution $\p$, and the goal is to estimate a symmetric property $f(\p)$ up to a $\pm\eps$ factor, with probability at least 2/3.
\paragraph{Quantum state property estimation.} The problem of estimating von Neumann and R\'enyi entropies of a quantum state can be shown to be equivalent to estimating a symmetric property $f(\p)$ of some distribution. However, instead of being given independent samples $X_1, \ldots, X_\ns$ from the distribution $\p$ as in the classical case, we are given access to a function $\lamb(X^\ns)=\lambi{1}\ge\lambi{2}\ge\ldots$ of $X^{\ns}$. Here $\lambi{1}, \lambi{2},\ldots$ are integers satisfying the following property.
\begin{itemize}
	\item For any $k\ge1$, $\sum_{i=1}^{k} \lambi{i}$ is equal to the largest possible sum of the lengths of $k$ disjoint non-decreasing subsequences of $X^{\ns}$. 
\end{itemize}
Equivalently, we may view the observations as the output of 
the Robinson–Schensted–Knuth (RSK) algorithm applied to the sequence
$X^n$, instead of being $X^n$ itself. The reader is referred to~\cite{ODonnellW17} for more details on the procedure. The copy complexity of estimating quantum entropy turns out to be equivalent to the problem of estimating classical entropy when given access to $\lamb(X^\ns)$. A simple data processing of the form $p\to X^\ns\to \lamb(X^\ns)$ shows that the complexity of estimating a quantum state property is at least as hard as estimating the same property in the classical setting. 

\subsection{Organization}
The paper is organized as follows. In Section~\ref{sec:our-results}, we state our results, followed by a brief description of our tools in Section~\ref{sec:techniques}, and related work in Section~\ref{sec:related-work}. Section~\ref{sec:quant} gives a summary of the quantum set-up. Section~\ref{sec:preliminaries} provides the preliminary results needed for setting up the paper. In particular, Section~\ref{sec:wss} describes the optimal quantum measurement for the class of properties we are interested in. Section~\ref{sec:int-alpha} proves our bounds for integral order R\'enyi entropy. Section~\ref{sec:non-int} proves the upper bounds for non-integral orders, and Section~\ref{sec:lb-large} shows the lower bounds on the performance of the empirical estimator.

\input{our-results}

%% file: our-results.tex
\subsection{Our Results}
\label{sec:our-results}

We consider the following problem.
\begin{center}
 \framebox{
\begin{minipage}{13.5cm} $\Pi(f,\dims,\eps)$: Given a property $f$, and access to independent copies of a $\dims$-dimensional mixed state $\mst$ (e.g. output of some quantum experiment), how many copies are needed to estimate $f(\mst)$ to within $\pm\eps$?\footnotemark
\end{minipage}
}\footnotetext{We seek
  success with probability at least $2/3$, which can be boosted to
  $1-\delta$ by repeating the algorithm $O(\log(1/\delta))$ times and taking the median.}
\end{center}

We study the copy complexity of estimating the entropy of a mixed state of dimension $\dims$. The copy complexity, denoted by $C(\prp, \dims, \eps)$, is the minimum number of copies required for an algorithm that solves $\Pi(f,\dims,\eps)$. Copy complexity is defined precisely in Section~\ref{sec:property-estimation}. 

We will use the standard asymptotic notations. We will be interested in characterizing the dependence of  $C(S, \dims, \eps)$, and $C(S_\renprm, \dims, \eps)$, as a function of $\dims$ and $\eps$. We assume the parameter $\renprm$ to be a constant, and focus on only the growth rate as a function of $\dims$ and $\eps$. 

We will now discuss our results, which are summarized in Table~\ref{tab:results-one} and Table~\ref{tab:results-two}. 
For comparison purposes, it is useful to recall the copy complexity
of quantum tomography, in which the goal is to learn the entire density matrix $\rho$. The problem has been studied in various works using various distance measures; and up to poly-logarithmic factors, for the standard distance measures, the copy complexity depends quadratically on the dimension $\dims$. Namely, it is  $\tilde{O}(\dims^2)$.\footnote{We discuss the copy complexity of some other problems in related work (Section~\ref{sec:related-work}).}
Similar to the sample complexity of estimating R\'enyi entropies of classical distributions from samples, our bounds are also dependent on whether $\renprm$ is less than one, and whether it is an integer. (See Table I of~\cite{AcharyaOST17}, and Section~\ref{sec:classical} for the sample complexity in classical settings.) We organize our results as a function of $\renprm$ as follows. 

\paragraph{Integral $\renprm>1$.} We obtain our most optimistic 
and conclusive results in this case. In Theorem~\ref{thm:alpha-int}, we show that $\copycmpa{\renprm} =\Theta\Paren{\max\left\{\frac{\dims^{1-1/\renprm}}{\eps^2},\frac{\dims^{2-2/\renprm}}{\eps^{2/\renprm}} \right\}}$. We note that the lower bounds here hold for \emph{all estimators}, not just of the estimators used in the upper bound. Furthermore, these bounds are sub-quadratic in $\dims$, namely we can estimate the R\'enyi entropy of integral orders even before we have enough copies to perform full tomography. The upper bounds are established by analyzing certain polynomials from representation theory that are related to the central characters of the symmetric group. The main contribution is to analyze the variance of these estimators, for which we draw upon various results from Kerov's algebra. For the lower bound, we design the spectrums of two mixed states such that their R\'enyi entropy differ by at least $\eps$, but require a large copy complexity to distinguish between them. We use various properties of Schur polynomials and other properties of integer partitions~\cite{Macdonald98, HardyR18}. 

\begin{remark}
The first term in the complexity dominates when $\eps<1/\sqrt{\dims}$, and is identical to the sample complexity of estimating R\'enyi entropy in the classical setting.
\end{remark}

\paragraph{$\renprm<1$.}
We analyze the Empirical Young Diagram (EYD) algorithm~\cite{AlickiRS88, KeylW01} for estimating $\rentprmmst$ for $\renprm<1$. The EYD algorithm is similar to using a plug-in estimate of the empirical distribution to estimate properties in classical distribution property estimation. We show that $C(S_\renprm, \dims, \eps) = O(\dims^{2/\renprm}/\eps^{2/\renprm})$. Since $\renprm<1$, this growth is faster than quadratic, namely the EYD algorithm requires more copies than is required for tomography. We complement this result by showing that in fact the EYD algorithm requires $\Omega(\dims^{1+1/\renprm}/\eps^{1/\renprm})$ copies, showing that the super-quadratic dependence on $\dims$ is necessary for the EYD algorithm. The upper bound is proved in Theorem~\ref{thm:non-int-upper-small-alpha}, and the lower bound in Theorem~\ref{thm:lb-non-int-small}. In comparison, in the classical setting the exponent of $\dims$ is almost $1/\renprm$. 

\paragraph{von Neumann entropy, $\renprm=1$.} Again using the EYD algorithm, in Theorem~\ref{thm:von-neumann-empirical} we show that $\copycmph = O(\dims^2/\eps^2)$. We formulate an optimization problem whose solutions are an upper bound on the bias of the empirical estimate, and we bound the variance by proving that the estimator has a small bounded difference constant. In Theorem~\ref{thm:lb-non-int-large} we show a lower bound of $\Omega(\dims^2/\eps)$ for the EYD estimator to estimate the entropy of the maximally mixed state. This complexity is still similar to that of full quantum tomography. 

\paragraph{Non integral $\renprm>1$.}
Again using the EYD algorithm, in Theorem~\ref{thm:non-int-upper-large-alpha}, we show that $C(S_\renprm, \dims, \eps) = O(\dims^{2}/\eps^2)$. We also provide a lower bound of $\Omega(\dims^2/\eps)$ for the EYD estimator in Theorem~\ref{thm:lb-non-int-large}. 

In addition to these results, we improve the error probability of the lower bounds on the convergence of EYD algorithm to the true spectrum. In particular, for the uniform distribution~\cite{ODonnellW15} have shown that unless the number of copies is at least $\Omega(\dims^2/\eps^2)$ the EYD has a total variation distance of at least $\eps$ with probability at least 0.01. We show that in fact unless the number of copies is at least $\Omega(\dims^2/\eps^2)$ the trace distance is at least $\eps$ with probability at least $1-\exp(c\cdot \dims^2)$ for some constant $c$. 

\begin{table}[t]
\renewcommand{\arraystretch}{1.3}
\begin{minipage}{.45 \linewidth} 
\renewcommand{\arraystretch}{1.3}
\begin{tabular}{|c|c|c|}\hline
	Upper Bound & Lower Bound\\ \hline
$O\left(\max\left\{\frac{\dims^{2-\frac2\renprm}}{\eps^{\frac2\renprm}}, \frac{\dims^{1-\frac1\renprm}}{\eps^2}\right\}\right)$ & $\Omega\left(\max\left\{\frac{\dims^{2-\frac2\renprm}}{\eps^{\frac2\renprm}}, \frac{\dims^{1-\frac1\renprm}}{\eps^2}\right\}\right)$\\ \hline	
\end{tabular}
\caption{Copy complexity of $\rentprmmst$ for integral $\renprm>1$.}
\label{tab:results-one}
\end{minipage}
\ \ \ \ \ \ \ \ \ \ \ \ \ 
\begin{minipage}{.3\linewidth}
      \begin{tabular}{| c | c | c |}\hline
      $\renprm$ & Upper Bound & Lower Bound\\ \hline
$\renprm>1$ & $O(\dims^2/\eps^2)$ & $\Omega(\dims^2/\eps)$\\ \hline
$\renprm <1$ & $O(\dims^{2/\renprm}/\eps^{2/\renprm})$ & {$\Omega(\dims^{1+1/\renprm}/\eps^{1/\renprm})$}\\ \hline
$\renprm=1$ & $O(\dims^2/\eps^2)$ &  $\Omega(\dims^2/\eps)$   \\ \hline
      \end{tabular}
    \caption{Copy complexity of empirical estimators.}
          \label{tab:results-two}
   \end{minipage}
\end{table}

%% file: our-techniques.tex
\subsection{Our Techniques}
\label{sec:techniques}

In this section, we provide a high level overview of the technical contributions of our paper. 

The entropy functions that we consider are unitarily invariant properties (Section~\ref{sec:unitarily}), namely they depend only on the multiset of eigenvalues of the density matrix. For example, a density matrix $\mst$ with eigenvalyes $\eigvi{1},\ldots, \eigvi{\dims}$, we have $S(\mst) = -\sum_i \eigvi{i}\log \eigvi{i}$, and $\rentprmmst = \log(\sum_i \eigvi{i}^{\renprm})/(1-\renprm)$, meaning that von Neumann, and R\'enyi entropy are unitarily invariant properties. For such properties, it is known that an optimal measurement scheme over the set of all measurements is the weak Schur Sampling (WSS) (Section~\ref{sec:wss}). The output of this measurement is a partition of $\lamb\vdash\ns$, usually denoted by a Young diagram (Section~\ref{sec:schur-power}), the number of independent copies of $\mst$ used. The goal is then to estimate the entropy from the output Young diagram supplied by WSS. 

Estimating R\'enyi entropy is equivalent to obtaining multiplicative estimates of the power sum $\Mapeigv{\renprm} \ed \sum \eigvi{i}^\renprm$. In the classical setting, it turns out that for integral $\renprm>1$, there are simple unbiased estimators of $\Mapeigv{\renprm}$. In the quantum setting, for integral $\renprm$, there are unbiased estimators for $\Mapeigv{\renprm}$. These estimators are now polynomials (called $p^{\#}$) over Young tableaus obtained from Kerov's algebra. While the estimator itself is simple to state in terms of $p^{\#}$ polynomials, bounding its variance requires a number of intricate arguments. Using results from representation theory about $p^{\#}$, we first write the variance of the estimator as a linear combination of $p^{\#}$ polynomials. We use combinatorial arguments about the cycle structure of compositions of permutations, and use that to show that only a certain subset of $p^{\#}$'s can appear in the variance expression. Moreover, the number of $p^{\#}$'s can be bounded using the Hardy-Ramanujam bounds on the partition numbers. We also provide bounds on the coefficients to finally obtain the upper bound for integral $\renprm$ (Theorem~\ref{thm:alpha-int}). 

For the lower bound for integral $\renprm$, one of the terms follows from the classical lower bounds, and the fact that estimation is easier in the classical setting than in the quantum setting. To prove a lower bound equal to the second term, we invoke the classical Le Cam's method combined with results on Schur polynomials and partition numbers.

Our upper bounds for von Neumann entropy and for non-integral $\alpha$ use the Empirical Young Diagram (EYD) algorithm (Section~\ref{sec:eyd}). This is akin to the empirical plug-in estimators for distribution property estimation. Our upper bounds require various bias and concentration results on the Young-tableaux. Fortunately, in the recent works of O'Donnell and Wright, a number of such bounds were proved. We build upon their results, and prove some additional results to show the copy complexity bounds for the EYD algorithm. 

To prove the lower bounds for the EYD algorithm, we design eigenvalues such that unless the number of copies is large enough, the EYD algorithm cannot concentrate around the true entropy. 

One of our contributions pertains to the convergence of the empirical Young diagram to the true distribution. A lower bound of $\dims^2/\eps^2$ was shown by~\cite{ODonnellW15}. However, their results only holds with a constant probability (with probability 0.01 to be precise). We show \emph{very sharp concentration} by invoking McDiarmid's inequality. We show that unless the number of samples is more than $\dims^2/\eps^2$ the empirical Young diagram's lower bound holds with probability $1-\exp(-c\dims)$ for some constant $c$. This exponential concentration result could be of independent interest.

%% file: related-work.tex
\subsection{Related Work}
\label{sec:related-work}

Our work is related to symmetric distribution property estimation in classical setting, property estimation of classical distributions using quantum queries, and the property estimation of quantum states (as in the set-up of this paper). We briefly mention some closely related works. The reader is encouraged to read the survey by Montanaro and de Wolf~\cite{MontanaroW13}, and the thesis by Wright~\cite{Wright16} for more details on the recent literature. 

\subsubsection{Symmetric Property Estimation of Discrete Distributions}
\label{sec:classical}

A property of a distribution is symmetric if it is a function of only the probability multiset. A number of properties, such as the Shannon entropy, R\'enyi entropy, KL divergence, support size, distance to uniformity, are all symmetric. While there is a long literature on some of these problems, the optimal sample complexity for these problems was established only over the last decade~\cite{Valiant11b, WuY14a, JiaoVHW15, AcharyaOST17, JiaoHW16, WuY15, OrlitskySVZ04, BuZLV16, HanJW14, AcharyaDOS17}. We mention the state of the art results, and the reader can consult the related papers and references therein to learn more about the landscape of symmetric distribution property estimation problems. 
Similar to the quantum setting, let $S(f,\dims, \eps)$ be the minimum number of samples needed from a discrete distribution $p$ over $\dims$ elements to estimate a property $f(p)$ up to $\pm\eps$, again with probability at least $2/3$.

For the Shannon entropy $H(p)= -\sum_x p(x)\log p(x)$, a long line of work culminated in~\cite{Valiant11b, WuY14a, JiaoVHW15} showing that $\smpcmph=\Theta\Paren{\frac{\dims}{\eps\log\dims}+\frac{\log^2 \dims}{\eps^2}}$.  

The problem of estimating R\'enyi entropy $H_\renprm(p)= \log\Paren{\sum_x p(x)^{\renprm}}/(1-\renprm)$, was studied in \cite{AcharyaOST14, AcharyaOST17, OmbreskiS17}. The sample complexity dependence in the classical setting seems to suggest the same qualitative behavior as our results. They show that for $\renprm<1$,  $\smpcmpa{\renprm} = O\Paren{\frac{\dims^{1/\renprm}}{\eps^{1/\renprm}\log\dims}}$, and for $\renprm>1, \renprm\notin\NN$, $\smpcmpa{\renprm} = O\Paren{\frac{\dims}{\eps^{1/\renprm}\log\dims}}$. Moreover, their information theoretic lower bounds show that the exponent of $\dims$ cannot be improved by any algorithm. For the case of integral $\renprm$, larger than one, they characterize the sample complexity up to constant factors by showing that $\smpcmpa{\renprm} = \Theta\Paren{\frac{\dims^{1-1/\renprm}}{\eps^2}}$. We note that this complexity is indeed one of the terms in our copy complexity for integral $\renprm$, which happens for large $\ns$.~\cite{OmbreskiS17} provide bounds that improve the sample complexity of R\'enyi entropy estimation, for distributions with small R\'enyi entropy. 

\subsubsection{Quantum Property Estimation of Mixed States}

While we are not aware of a lot of literature on property estimation of mixed states, there are now many works on the related problem of quantum property testing, where the goal is to find the copy complexity of deciding whether a mixed state has a certain property of interest, and on the problem of quantum tomography, where the goal is to learn the entire density matrix $\mst$. 

The copy complexity of quantum tomography is quadratic in $\dims$, and the complexity for tomography in various distance measures have been studied in~\cite{HaahHWWY17, OW16, OW17}. 

Testing whether $\mst$ has a particular unitarily invariant property of interest was studied in~\cite{ODonnellW15} for a number of properties. They show that for testing whether $\mst$ is maximally mixed, namely whether all elements of $\eigv$ are $1/\dims$, requires $\Theta(\dims/\eps^2)$ copies. They also studied the problem of testing the rank of $\mst$, and also provide bounds on the performance of the EYD algorithm for estimating the spectrum. {Recently,~\cite{BadescuOW17} obtained tight bounds on the copy complexity of testing whether an unknown density matrix is equal to a known density matrix.} The optimal measurement schemes for some of these problems can be computationally expensive. Testing properties under simpler \emph{local measurements} was studied recently in~\cite{PallisterML17}.

In a personal communication, Bavarian, Mehraban, and Wright~\cite{BMW16} claim an algorithm with copy complexity $O(\dims^2/\eps)$ for the von Neumann entropy estimation, which is an $\eps$ factor improvement over our bound. 

\subsubsection{Quantum Algorithms for Classical Distribution Properties}
Testing and estimating distribution properties using quantum queries has been considered by various authors. Problems of testing properties such as uniformity, identity, closeness under the regular quantum query model, and conditional quantum query models have been studied in~\cite{BravyiHH11, ChakrabortyFMW10, SardharwallaSJ17}. 

Recently Li and Wu~\cite{LiW17} studied the quantum query complexity of estimating entropy of discrete distributions. They provide bounds on the query complexity for estimating von Neumann entropy, and R\'enyi entropy. For certain values of $\renprm$, the bounds on query complexity can in fact be at times quadratically better than the corresponding sample complexity bounds.

%% file: formulation.tex
\section{Quantum Measurements and Property Estimation}
\label{sec:quant}
\subsection{Density Matrix and Quantum Measurement}
A quantum state is described by a \emph{density matrix} $\mst$, which is a $\dims$-dimensional positive semi-definite matrix with unit trace. The joint state of $\ns$ independent copies is given by the tensor product $\mst^{\otimes n}=\mst\otimes\cdots\otimes\mst$, which is a density matrix of dimension 
${d^n \times d^n}$.

\emph{Quantum measurements} are described by a set of matrices $\{M_m\}$ called measurement operators, where index $m$ denotes the measurement outcome. Measurement operators satisfy the completeness condition, $\sum_m{M_m^{\dagger}M_m}=I$. If the  pre-measurement state is $\mst$ then probability of measurement outcome $m$ is $\tr{M_m^{\dagger}M_m\mst}$, and the post-measurement state is $M_m\mst M_m^{\dagger}/\tr{M_m\mst M_m^{\dagger}}$. The measurement operators are also allowed to have an infinite outcome set, in which case a suitable $\sigma$-algebra on the set of outcomes and a probability measure on this space are defined. 
For a detailed discussion of these concepts see~\cite{nielsen2010quantum}.

\subsection{Property Estimation}
\label{sec:property-estimation}
A \emph{property} $\prp (\mst)$ maps a mixed state $\mst$ to $\RR$. 
Given $n$ and $d$, an \emph{estimator} is a set of measurement matrices 
$\{M_m\}_{m = 1}^\infty$ for the state space $\CC^{d^n \times d^n}$ and
a ``classical processor'' $g(\cdot)$, which maps the natural
numbers to $\RR$. Given $n$ copies of a state $\mst$, the
estimator proceeds by applying the measurement $\{M_m\}_{m = 1}^\infty$
to the state $\mst^{\otimes n}$
and then applying $g(\cdot)$ to the resulting outcome.
Given a property $\prp$, accuracy parameter $\eps$, error parameter $\delta$, and access to $\ns$ independent copies of a mixed state $\mst$, 
we seek an estimator $\fhat$ such that with probability at least $1-\delta$ 
\[
\absv{f(\mst)-\fhat(\rho^{\otimes n})}<\eps.
\] 

The \emph{copy complexity} of $\prp$ is  
\begin{align}
C(\prp, \dims, \eps, \delta) \ed \min\left\{\ns: \exists \fhat : \forall \mst, \fhat \text{ is a } \pm\eps \text{ estimate of } \prp(\mst) \text{ with probability} >1-\delta\right\},\nonumber
\end{align}
the minimum number of copies required to solve the problem. 
Throughout this paper we will consider $\delta$ to be a constant, say 1/3. We can boost the error to any $\delta$ by repeating the estimation task $O(\log(1/\delta))$ times, and taking the median of the outcomes. This causes an additional $O(\log(1/\delta))$ multiplicative cost in the copy complexity. We denote 
\begin{align}
C(\prp, \dims, \eps) \ed C(\prp, \dims, \eps, 1/3).
\end{align}

\subsection{Unitarily Invariant Properties}
\label{sec:unitarily}

Suppose $U(\dims)$ is the set of all $\dims\times\dims$ unitary matrices. 
\begin{definition}
A property $\prp(\mst)$ is called \emph{unitarily invariant}, if $\prp(U\mst U^{\dag}) = \prp(\mst)$ for all $U\in\unitary$. 
\end{definition}
Let $\eigv = \{\eigvi{1}\upto\eigvi{\dims}\}$ be the multiset the eigenvalues (also called as spectrum) of $\mst$. Two density matrices $\rho$, and $\sigma$ have the same spectrum if and only if there is a unitary matrix $U$ such that $\sigma = U\rho U^{\dag}$. Therefore, unitarily invariant properties are functions of only the spectrum of the density matrix. Since density matrices are positive semi-definite with unit trace, then $\sum_i \eigvi{i}=1$, and we can view $\eigv$ as a distribution over some set. Unitarily invariant properties are analogous to properties in classical distributions that are a function of only the multiset of probability elements, called symmetric properties. 

For a density matrix with eigenvalues $\eigvi{1},\ldots, \eigvi{\dims}$, we have $S(\mst) = -\sum_i \eigvi{i}\log \eigvi{i}$, and $\rentprmmst = \log(\sum_i \eigvi{i}^{\renprm})/(1-\renprm)$. Quantum entropy can be viewed as the classical entropy of the distributions defined by $\eigv$, and in particular they are unitarily invariant.  

Working with unitarily invariant properties is greatly simplified by the following powerful result~\cite{KeylW01, ChildsHW07, Harrow05, Christandl06} (See~\cite[Section 4.2.2]{MontanaroW13} for details). 
\begin{lemma}
A quantum measurement called \emph{weak Schur sampling} is optimal for estimating unitarily invariant properties. 
\end{lemma}
\noindent Weak Schur sampling is discussed in Section~\ref{sec:wss}.

%
%
%

%% file: preliminaries.tex
\section{Preliminaries}
\label{sec:preliminaries}
We list some of the definitions and results we use in the paper.

\begin{definition}
The \emph{total variation} distance, \emph{KL divergence}, and \emph{$\chi^2$ distance} between distributions $\p$, and $\q$ over $\cX$ are 
\begin{align}
\dtv{\p}{\q} \ed &\sup_{A\subset\cX} \p(A) - \q(A) = \frac12\|\p-\q\|_1, \\
d_{KL}(\p,\q) \ed &\sum_{x\in\cX} \p(x)\log\frac{\p(x)}{\q(x)},\\
\chisq{\p}{\q} \ed &\sum_{x\in\cX} \frac{(\p(x)-\q(x))^2}{\q(x)}.
\end{align}
\label{def:distances}
\end{definition}
The distance measures satisfy the following bound.
\begin{lemma}
\label{lem:distance-bounds}
\[2\dtv{\p}{\q}^2 \le d_{KL}(\p,\q) \le \chisq{\p}{\q}.\]
\end{lemma}
\noindent The first inequality is Pinsker's Inequality, and the second follows from concavity of logarithms. 

We now state some concentration results that we use. 

Let $f:[\dims]^{\ns}\to\RR$ be a function, such that
\begin{align}
\max_{z_1,\ldots, z_n, z_{i}^{'}\in[\dims]} \absv{f(z_1,\ldots,z_n) - f(z_1,\ldots, z_{i-1}, z_i^{'},\ldots,z_n)}\le c_i,\label{eqn:conc}
\end{align}
for some $c_1, \ldots, c_\dims$. 

The next two results show concentration results of functions $f$ that satisfy~\eqref{eqn:conc}. The following lemma is ~\cite[Corollary~3.2]{BoucheronLM13}.
\begin{lemma} 
\label{lem:bdd-diff-var}
For independent variables $Z_1,\ldots,Z_{\ns}$, 
\[
\text{Var}\Paren{{f(Z_1,\ldots,Z_{\ns})}} \le \frac14\sum_{i=1}^{\ns} c_i^2.
\]
\end{lemma}

The next result is McDiarmid's inequality~\cite[Theorem~6.2]{BoucheronLM13}. 
\begin{lemma}
\label{lem:mcdiarmid}
For independent variables $Z_1,\ldots,Z_{\ns}$, 
\[
\probof{{f(Z_1,\ldots,Z_{\ns})-\expectation{f(Z_1,\ldots,Z_{\ns})}}>t}\le  e^{-\frac{2t^2}{c_1^2+\ldots+c_{\ns}^2}}.
\]
\end{lemma}

%% file: polynomials.tex
\subsection{Schur Polynomials and Power-Sum Polynomials}
\label{sec:schur-power}

A \emph{partition} $\lamb$ of $\ns$ is a collection of non-negative 
integers $\lambi{1}\ge\lambi{2}\ge\ldots$ that sum to $\ns$. 
We write $\lamb\vdash\ns$ and we write
$\Lambda_n$ for the set of all partitions of $n$.
We denote the number of positive
integers in $\lamb$ by $\ell(\lamb)$, which we call its \emph{length}.
An partition $\lamb$ can be depicted
with an English \emph{Young diagram}, which consists of a row of 
$\lambda_1$ boxes above a row of $\lambda_2$ boxes, etc., as 
showed in Fig.~\ref{Young_Diagram}. The partition associated with a Young
diagram is called its \emph{shape}.
Note that the number of rows in the
Young diagram of $\lamb$ is $\ell(\lamb)$ and the total number
of boxes is $n$.  A \emph{Young tableau} over alphabet $[d]$ 
is a Young diagram
in which each box has been filled with an element of $[d]$. A 
Young tableau is called \emph{standard} if it is strictly 
increasing left-to-right across each row and top-to-bottom
down each column. A Young tableau is \emph{semistandard}
if it is strictly increasing top-to-bottom down each column
and nondecreasing left-to-right across each row. Given 
$\lamb\vdash\ns$ and $d$, the \emph{Schur polynomial} 
is the polynomial in the variables $x_1, x_2, \ldots, x_d$ 
defined by
\begin{equation}
\label{eq:schurdef}
s_\lamb(x) = \sum_{T} \prod_{i = 1}^d x_i^{\#(T,i)},
\end{equation}
where the sum is over the set of all semistandard Young Tableaus
over alphabet $[d]$
corresponding to the partition $\lamb$ and $\#(T,i)$ is the number
of times $i$ appears in $T$. Schur polynomials turn out to be
symmetric, meaning that they are invariant to the ordering of the
variables $x_1,\ldots,x_d$~\cite{Macdonald98, Stanley}.

\begin{figure}
\centering
\includegraphics[scale=0.25]{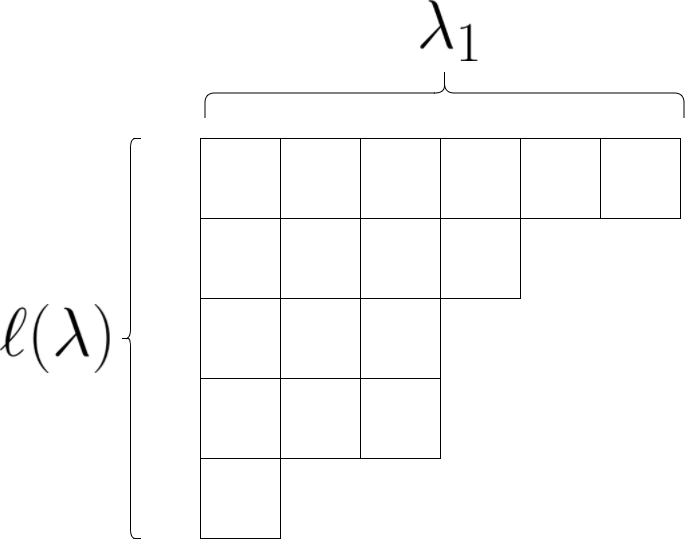}
\caption{English Young diagram for the partition $\lamb=(6,4,3,3,1)$.}
\label{Young_Diagram}
\end{figure}

We shall also consider polynomials obtained from power sums.
Given $\renprm\in\RR_{\ge0}$ and a distribution $\eigv$ on $[\dims]$,\footnote{Power sums can are usually defined for general vectors. We will consider them only for distributions in this paper.} define
\begin{align*}
\Map{\renprm}{\eigv} \ed \sum_{i=1}^{\dims} \eigv_i^{\renprm}.
\end{align*}
Given $\lamb\vdash r$, we define the \emph{power sum} polynomial
by 
$$
\Map{\lamb}{\eigv} = \prod_{i = 1}^{\ell(\lamb)} \Map{\lambda_i}{\eigv}
$$
The following is Lemma 1 in~\cite{AcharyaOST17}, which describes a number of inequalities that hold for the power sums of distributions.
\begin{lemma}\label{lem:bnd_moments}
Suppose $\eigv$ is a distribution over $\dims$ elements, then
\begin{enumerate}[label=(\roman*)]
\item
For $\renprm <1$,
\[
1\leq \Mapeigv\renprm \leq \dims^{1-\renprm}, 
\]
and for $\renprm>1$,
\[
\dims^{1-\renprm} \leq \Mapeigv \renprm \leq 1.
\]
\item
For every $\renprm\geq 0$ and $\beta \geq 0$,
\[
\Mapeigv {{\renprm+\beta}}\le \Mapeigv{\renprm}^{{\frac{\renprm+\beta}{\renprm}}}.
\]
\item
 For every $\renprm\geq 0$,
\[
\Mapeigv{2\renprm} \leq {\Mapeigv{\renprm}}^2.
\]
\item
For $\renprm>0$ and $ 0\leq \beta \leq \renprm$,
\[
\Mapeigv{\renprm - \beta} \leq \dims^{\frac \beta\renprm}\, \Mapeigv{\renprm}^{\frac{\renprm -\beta}\renprm}.
\]
\item
For $\renprm \geq 1$ and $ 0\leq \beta \leq \renprm$,
\[
\Mapeigv {\renprm+\beta}\le \dims^{(\renprm-1)(\renprm-\beta)/\renprm}\, \Mapeigv {\renprm}^2,
\]
and
\[
\Mapeigv {\renprm -\beta}\le \dims^\beta\, \Mapeigv {\renprm}.
\]
\end{enumerate}
\end{lemma}

Schur polynomials and power-sum polynomials are related through 
a change of basis.
There exists a function $\chi_\cdot(\cdot):
\Lambda_n^2 \mapsto \mathbb{R}$ such that~\cite[Theorem~7.17.3]{Stanley}
\begin{align}
\label{eqn:power-central}	
M_{\boldsymbol{\mu}}(\cdot) = 
  \sum_{\lamb} \chi_{\lamb}(\boldsymbol{\mu}) s_\lamb(\cdot).
\end{align}
The $\chi_\cdot(\cdot)$ function in fact comprises the characters of
the irreducible representations of the symmetric group on
$[n] = \{1,\ldots,n\}$~\cite[Sec.~7.18]{Stanley}, 
although this fact is not needed. The 
$\chi_\cdot(\cdot)$ function can also be defined 
combinatorially~\cite{Stanley}. The quantity
$\chi_{\lamb}(\boldsymbol{\mu})$ is difficult
to compute in general~\cite{Hepler94}, although we shall only
be interested in particular $\boldsymbol{\mu}$, as follows.
Let $\text{dim}(\lamb)$ denote the number of standard Young
tableaus over alphabet $[n]$ with shape $\lamb$. For 
$\lamb\vdash\ns$ and $\boldsymbol{\mu}\vdash r$ define
\[
\proh{\boldsymbol{\mu}}{\lamb} \ed\begin{cases}
\flnpwr{n}{r}\cdot \frac{ \chi_{\lamb} ({\boldsymbol{\mu}\cup1^{n-r}}) }{\text{dim}(\lamb)} & \text{ if } n\ge r,\\
0 & \text{otherwise}.
\end{cases}
\]  
where $\flnpwr{n}{r}$ is the \emph{falling power}, i.e.,
$\flnpwr{n}{r} = n\cdot(n-1)\cdot(n-2)\cdot(n - (r-1))$ and
$\boldsymbol{\mu}\cup1^{n-r}$ denotes the partition of $[n]$ 
consisting of $\boldsymbol{\mu}$ followed by $n-r$ ones.


%% file: wss.tex
\subsubsection{Weak Schur Sampling (WSS)}
\label{sec:wss}

We describe some of the key results about weak Schur sampling (WSS) that we will use in this paper. The readers can refer to~\cite[Section 4.2.2]{MontanaroW13},~\cite[Chapter 3]{Wright16}, and references therein for further details. 

Weak Schur Sampling is a measurement that takes $n$ independent copies of a mixed state $\mst$ (denoted  $\mst^{\otimes\ns}$), and outputs a $\lamb\vdash\ns$. The output distribution over partitions is called Schur-Weyl distribution, denoted $SW_{\eigv}$, and the probability of $\lamb\vdash\ns$ is given by
\begin{equation} 
\label{eq:WSSdist}
SW_{\eigv}(\lamb) = \text{dim}(\lamb)\cdot \schurlx{\lamb}{\eigv},
\end{equation}
where, recall from the previous section that $\text{dim}(\lamb)$ is the number of Standard Young Tableaux of shape $\lamb$, and $\schurlx{\lamb}{\eigv}$ is the Schur polynomial with variables $\eigv$, and shape $\lamb$. Since Schur polynomials are symmetric, this probability is only a function of the multiset of eigenvalues, namely a function of the eigenvalue spectrum. 

An alternate combinatorial characterization of the output of WSS is given next. Some of the intermediate steps involving the Robinson-Schensted-Knuth (RSK) correspondences, and Green's theorem are not invoked later in the paper, and are omitted. We simply describe the method by which the final diagram is obtained. The reader can refer to the short survey~\cite{ODonnellW17} for details on the combinatorial procedure. 

\noindent Suppose $\mst$ is a mixed state with the multiset of eigenvalues $\{\eigvi{1}, \ldots, \eigvi{\dims}\}$. 
\begin{enumerate}
\item 
Consider a distribution over $[\dims]$, where $i$ has probability $\eigvi{i}$. 
\item
Draw $\Xon$ independently from this distribution. 
\item
Let $\lamb = \lambi1\ge \lambi2\ge \ldots$, be such that for any $k>0$, $\lambi1+\ldots+\lambi k$ is equal to the \emph{largest sum of lengths} of $k$ disjoint non-decreasing subsequences of $\Xon$. 
\end{enumerate}
The output distribution of this process is the same as that of weak Schur Sampling~\cite{Wright16}. Furthermore, {one of the results proved in~\cite{Wright16} is that the output distribution of the procedure above is \emph{independent} of the ordering of $\eigvi{i}$'s, and only depends on the multiset of the eigenvalues. For example, when $\dims=2$, the distributions $\eigvi{1}=0.2, \eigvi{2}=0.8$, and the distribution $\eigvi{1}=0.8, \eigvi{2}=0.2$ have the same output distributions over Young tableaux generated by the procedure above.}

Since the Young tableaux is a function of the sequence generated by the spectrum distribution, 
\begin{lemma}
\label{lem:class-quant}
The copy complexity of estimating a unitarily invariant property of a mixed state is at least the sample complexity of estimating the same symmetric property of the spectrum distribution. 
\end{lemma}

The $\proh{\lamb}{\boldsymbol\mu}$ polynomial defined in the last section is useful to us due to the
following lemma, which states that the (normalized) polynomial $\phash{\lamb} (r)$ is an unbiased estimator of the $r$th moment of $\eigv$. The lemma follows from the definitions and results already mentioned, and is implicit in~\cite{meliot2010kerov,IvanovK02}, and explicit in~\cite[Proposition 3.8.3]{Wright16}.  
\begin{lemma} Fix a distribution $\eigv$, a natural number $r$, and
any partition $\boldsymbol{\mu}$ of $r$. If $\lamb$ is randomly
generated according to the distribution in~(\ref{eq:WSSdist}) then
\begin{align}
\EE\Brack{\proh{\boldsymbol{\mu}}{\lamb}} = \flnpwr{n}{r} \cdot \Map{\boldsymbol{\mu}}{\eigv} =\flnpwr{n
}{r} \cdot \prod_{i} {\Map{\mu_i}{\eigv}}.\end{align}
In the special case when $\boldsymbol{\mu} = (r)$, a partition with only one part, we have
\begin{align}
\expectation{\phash{(r)} (\lamb) } &= \flnpwr{\ns}{r}\Map{r}{\eigv}.
\end{align}
\label{lem:power-sums}
\end{lemma}
\begin{proof}
Plugging in the probability of $\lamb$ from~\eqref{eq:WSSdist}, and the definition of $\phash{(r)} (\lamb) $ from Section~\ref{sec:schur-power}, and finally using~\eqref{eqn:power-central} gives the lemma.
\end{proof}

\subsubsection{The EYD algorithm, and classical plug-in estimation}
\label{sec:eyd}
The EYD algorithm is a simple algorithm for estimating $f(\mst)$. The algorithm works in two steps. 
\begin{itemize}
\item 
Compute the empirical distribution, which assigns probability $\lambi{i}/\ns$ to the symbol $i$. 
\item
Output the property $f$ of a mixed state with eigenvalues equal to $\lambi{i}/\ns$. 
\end{itemize}

The EYD algorithm is a quantum analogue of the classical empirical/plug-in estimator, which works as follows. Consider the step 2 of the weak Schur sampling procedure explained in Section~\ref{sec:wss}, which generates $X^\ns$,  $\ns$ \textit{i.i.d.} samples from the distribution $\eigv$ over $\{1,\upto\dims\}$. Let $\hat{p}$ be the empirical distribution of $X^n$, which assigns a probability $\hat{p}(i)=N_i/\ns$ to a symbol $x$, where $N_i$ is the number of times symbol $i$ appears in $X^n$. The plug-in estimator, upon observing $X^\ns$, outputs $f(\hat{p})$. The plug-in estimator has been widely studied in statistics literature.

An observation from the non-decreasing subsequence interpretation of the weak-Schur sampling is that for any sequence $X^\ns$, the distribution $\lambi{i}/\ns$ majorizes the corresponding empirical distribution. This follows from the fact that the length of longest $k$ disjoint non-decreasing sub-sequences is always at least the sum of the $k$ largest $N_i$'s. In particular, we can state the following result.
\begin{lemma}
\label{lem:emp-eyd}
Consider the sorted plug-in distribution $\hat{p}$ of $X^n$, and the distribution $\lambi{i}/\ns$ obtained from $X^\ns$ by the WSS procedure. $\lamb/\ns$ majorizes $\hat{p}$, namely, for all $j$, $\sum_{i=1}^{j} \lambi{i}/\ns\ge \sum_{i=1}^{j} \hat {p}(i)$.
\end{lemma}

%
%
%

%% file: upper-bound-techniques.tex
\subsection{Proving Upper Bounds on Copy Complexity}

%

Consider $\renprm\ne 1$ and $\hat{\eps} \in (0,1)$. Suppose $\widehat{\Mapeigv{\renprm}}$ satisfies 
\begin{align}
\absv{\widehat{\Mapeigv{\renprm}}-\Mapeigv{\renprm}}\le \hat{\eps} \Mapeigv{\renprm}.\nonumber
\end{align}
Then 
\begin{align}
\absv{\frac1{1-\renprm}\log \widehat{\Mapeigv{\renprm}}-\rentprmmst}
& = \absv{\frac1{1-\renprm}\log \frac{\widehat{\Mapeigv{\renprm}}}{{\Mapeigv{\renprm}}}}\nonumber\\
& \leq \absv{ \frac1{1-\renprm} \max \left\lbrace \log (1+\hat{\eps}), \absv{\log(1-\hat{\eps})} \right\rbrace } \nonumber \\
& = \absv{ \frac{\log(1-\hat{\eps})}{1-\renprm}}.
\end{align}
Therefore, to obtain a $\pm \eps$ estimate of $\rent{\renprm}{\eigv}$, it suffices to derive a $1-e^{-\eps \absv{1-\renprm}}$ multiplicative estimate of $\Mapeigv{\renprm}$. Note that $1-e^{-\eps \absv{1-\renprm}} \geq \frac{\eps \absv{1-\renprm}}{1+\eps \absv{1-\renprm}}$ since $e^{-x} \leq \frac{1}{x+1}$ for $x > -1$. Moreover, in the regime in which $\epsilon$ does not grow with $\dims$, $\frac{\eps \absv{1-\renprm}}{1+\eps \absv{1-\renprm}} = \theta(\epsilon)$. Therefore, in the remainder of the paper, we will be interested in $1+\epsilon$ multiplicative estimators.

Finally note that for any $X$, 
\[
\probof{\absv{X-\EE[X]}^2> 9 \cdot \Var(X)}<\frac 19
\]
by Markov's inequality. Since $\EE\Brack{\proh{(\renprm)}{\lambda}} = \flnpwr{\ns}{\renprm} \Mapeigv{\renprm}$ (by Lemma~\ref{lem:power-sums}), then we get an $1+\eps$ multiplicative estimator of $\Mapeigv{\renprm}$ with probability at least 8/9 if
\begin{align}
\Paren{ \eps \cdot\EE\Brack{\proh{(\renprm)}{\lambda}}}^2\ge 9\cdot \Var\Paren{\proh{(\renprm)}{\lambda}}.\label{eqn-var}
\end{align}

%% file: integral-alpha.tex
\section{Measuring $\rentprmmst$ for integral $\renprm$}
\label{sec:int-alpha}
Our main result for integral $\renprm \ge 1$ is the following tight bound (up to constant factors)  on the copy complexity of estimating $\rentprmmst$. 
\begin{theorem} For $\renprm \in \mathbb{N}\backslash{\{1\}}$,
\label{thm:alpha-int}
\[
\copycmpa{\renprm}= \Theta\Paren{\max\left\{\frac{\dims^{1-1/\renprm}}{\eps^2},\frac{\dims^{2-2/\renprm}}{\eps^{2/\renprm}} \right\}},
\]
where the hidden constants depend only on $\renprm$.
\end{theorem}

%% file: achievability-integral.tex
\subsection{Achievability}

Our Renyi entropy estimator is simple, and is described in Algorithm~\ref{alg:renyi}.
\begin{algorithm}[htb]	

  	\begin{algorithmic}[1] 
	\State \textbf{Input:} $\ns$ independent copies of the state $\rho$, and $\renprm\in\NN$
	\State Run weak Schur sampling to obtain $\lamb\vdash\ns$. 
	\State Let $(\renprm)$ be the partition of $\renprm$ with one part.
	\State Compute $\proh{(\renprm)}{\lambda} = \flnpwr{n}{\renprm}\cdot \frac{\chi_{(\renprm)\cup1^{n-\renprm}}^{\lambda}}{\text{dim}(\lambda)}$.
	
	\State\textbf{Output:} $\frac1{1-\renprm}\log\Paren{\frac{\proh{(\renprm)}{\lambda}}{\flnpwr{n}{\renprm}}}$.
  	\end{algorithmic}

	\caption{Estimating Renyi entropy for integral $\renprm$'s.}
	\label{alg:renyi}	
\end{algorithm}

Note that we could have simply removed the $\flnpwr{n}{\renprm}$ terms from the algorithm's description, but these polynomials have a number of applications in representation theory to study the Symmetric group, and we simply keep the notation and definitions intact. 

To prove the theorem, we bound the expectation and concentration of $\proh{(\renprm)}{\lambda}$.
\begin{lemma} There is a constant $C_\renprm$ depending only on $\renprm$ such that
\begin{align}
\expectation{\phash{(\renprm)}(\lamb)} &= \flnpwr{\ns}{\renprm}\Map{\renprm}{\eigv}\label{eqn:expectation}\\
\variance{\phash{(\renprm)}(\lamb)} &\le C_\renprm\cdot \ns^{\renprm}\Paren{1+ \ns^{\renprm-1}\Mapeigv {2\renprm-1}}\label{eqn:hope-variance}.
\end{align} \label{lem:expectation-variance}
\end{lemma} 
\subsubsection{Proof of Theorem~\ref{thm:alpha-int} using Lemma~\ref{lem:expectation-variance}}

We want~\eqref{eqn-var} to hold, which happens if 
\begin{align} \label{eqn-var-2}
\Paren{\eps \flnpwr{\ns}{\renprm}\Map{\renprm}{\eigv}}^2\ge 9 C_\renprm\cdot \ns^{\renprm}\Paren{1+ \ns^{\renprm-1}\Mapeigv {2\renprm-1}}.
\end{align}
We claim that $ n=\Theta\Paren{\max\left\{\frac{\dims^{1-1/\renprm}}{\eps^2},\frac{\dims^{2-2/\renprm}}{\eps^{2/\renprm}} \right\}}$ is sufficient for~\eqref{eqn-var-2} to hold. Note that $\flnpwr{\ns}{\renprm}=\Theta(\ns^\renprm)$ and let $\tilde{c}_{\renprm} > 0$ be such that $\flnpwr{\ns}{\renprm} \geq \sqrt{\tilde{c}_{\renprm}} \ns^{\renprm}$. Now suppose $n \geq c_\renprm \max\left\{\frac{\dims^{1-1/\renprm}}{\eps^2},\frac{\dims^{2-2/\renprm}}{\eps^{2/\renprm}} \right\}$ for some constant $c_\renprm \geq \max \{ (18C_\renprm/ \tilde{c}_{\renprm} )^{1/\renprm}, 18 C_\renprm/ \tilde{c}_{\renprm} \} $. Then 
\begin{align}
\Paren{ \eps \flnpwr{\ns}{\renprm}\Map{\renprm}{\eigv}}^2 
& \geq \tilde{c}_{\renprm} \eps^2 \ns^{2 \renprm} \Map{\renprm}{\eigv}^2 \notag \\
&  \geq \frac{ \tilde{c}_{\renprm} }{2} \eps^2 \ns^{2 \renprm} \left( \frac{1}{\dims^{2\renprm-2}} +\frac{\Mapeigv{2\renprm-1}}{d^{1-1/\renprm}} \right) \label{eqn:stepone} \\
&  \geq  \frac{ \tilde{c}_{\renprm} }{2} \eps^2 \ns^{2 \renprm} \left( \frac{c_\renprm^\renprm}{\eps^2 \ns^\renprm} + \frac{c_\renprm \Mapeigv{2\renprm-1}}{\ns \eps^2}\right) \label{eqn:steptwo} \\
& = \frac{ \tilde{c}_{\renprm} }{2}  \ns^{ \renprm} \left( c_\renprm^{\renprm} + c_\renprm \ns^{\renprm-1}\Mapeigv{2 \renprm-1} \right)\nonumber \\
& \geq 9 C_\renprm\cdot \ns^{\renprm}\Paren{1+ \ns^{\renprm-1}\Mapeigv {2\renprm-1}},\nonumber
\end{align}
where~\eqref{eqn:stepone} follows from the fact that $\Mapeigv{\renprm}\ge\dims^{1-\renprm}$ and $\Mapeigv{2\renprm-1}\le \dims^{1-1/\alpha}\Mapeigv{\renprm}^2$ (by Lemma~\ref{lem:bnd_moments} ($i$) and ($v$)), and~\eqref{eqn:steptwo} follows from the assumption that $ \ns \geq c_\renprm \frac{\dims^{1-1/\renprm}}{\eps^2}$ and $ \ns \geq c_\renprm \frac{\dims^{2-2/\renprm}}{\eps^{2/\renprm}}$.


%

\subsubsection{Proof of Lemma~\ref{lem:expectation-variance}}
Equation \eqref{eqn:expectation} has already been established (cf. Lemma~\ref{lem:power-sums}). It remains to bound the variance of the estimator. 
\begin{align}
\Var\Paren{\proh{(\renprm)}{\lambda}} = \expectation{\proh{(\renprm)}{\lambda}^2} -\expectation{\proh{(\renprm)}{\lambda}}^2.\nonumber
\end{align}
The second term is evaluated from the means of the $\proh{\renprm}{\lambda}$ polynomials, which we know. For the first term, we need to bound the expectation of the products of such polynomials. In fact, there is a general result~\cite[Proposition 4.5]{IvanovK02}\cite[Corollary 3.8.8]{Wright16} that states that for any $\mu_1, \mu_2$, 
\[
\proh{\mu_1}{\lambda}\cdot \proh{\mu_2}{\lambda} = \proh{\mu_1\cup\mu_2}{\lambda} + \text{linear combination of $p^{\#}$'s for partitions of size at most $|\mu_1\cup\mu_2|-1$.}
\]
In our case, both the partitions are $(\renprm)$. So we can write 
\[
\proh{(\renprm)}{\lambda}\cdot \proh{(\renprm)}{\lambda} = \proh{(\renprm)\cup(\renprm)}{\lambda} + \sum_{\mu \in \cS} C_{\mu}\proh{\mu}{\lambda}, 
\]
where $C_{\mu}$ is at most $(\alpha!)^{2} < \exp(O(\renprm \log\renprm))$, and $\cS$ is the set of all partitions $\mu$ that can be obtained through the following procedure: 
\begin{enumerate}
\item 
Let $j$ be an integer in the set $\{0\upto\renprm-1\}$. 
\item
Let $\sigma_1$ be a permutation over $[\renprm+j]$ that has a cycle over the elements $\{1\upto\renprm\}$, and all the remaining elements are fixed points (the set $\{\renprm+1\upto\renprm+j\}$ for $j \geq 1$).
\item
Let $\sigma_2$ be a permutation over $[\renprm+j]$ that has a cycle over the elements $\{j+1\upto j+\renprm\}$, and all the remaining elements are fixed points (the set $\{1 \upto j\}$ for $j \geq 1$).
\item 
Let $\mu$ be the cycle structure of $\sigma_1 \circ \sigma_2$.
\end{enumerate}

The set of partitions that can be obtained through the above procedure for a \emph{fixed} $j \in \{0 \upto \renprm-1\}$ will be denoted by $\cS_j$. Now consider,
%
%
\begin{align}
\Var\Paren{\proh{(\renprm)}{\lambda}} 
=& \expectation{\proh{(\renprm)}{\lambda}^2} -\expectation{\proh{(\renprm)}{\lambda}}^2\nonumber\\
=& \expectation{ \proh{(\renprm,\renprm)}{\lambda} + \sum_{\mu\in\cS}C_{\mu}\proh{\mu}{\lambda}}-
\expectation{\proh{(\renprm)}{\lambda}}^2\nonumber\\
=& \flnpwr{n}{2\renprm}\Mapeigv{(\renprm,\renprm)} + \sum_{\mu\in\cS}C_{\mu}\flnpwr{n}{\absv{\mu}} \Mapeigv{\mu}-
\Paren{\flnpwr{n}{\renprm}\Mapeigv{\renprm}}^2\nonumber\\
=& \Paren{\flnpwr{n}{2\renprm}-(\flnpwr{n}{\renprm})^2} \Mapeigv{\renprm}^2 + \sum_{\mu\in\cS}C_{\mu}\cdot \flnpwr{n}{\absv{\mu}} \Mapeigv{\mu},\nonumber
\end{align}
where we have used that $\Mapeigv{\renprm}^2 = \Mapeigv{(\renprm,\renprm)}$. To bound $\Mapeigv{\mu}$ for $\mu \in \cS$, we use the following two lemmas. Lemma~\ref{lem:lengthpartition} is proved in Appendix~\ref{app:lemma-partition} and Lemma~\ref{lem:majorization} is proved in Appendix~\ref{app:lemma-majorization}. Recall that for a partition $\mu$, $\ell(\mu)$ denotes the length of the partition.

\begin{lemma} \label{lem:lengthpartition}
For all $j \in \{0 \upto \renprm-1\}$ and $\mu \in \cS_j$, $\ell(\mu)\le\renprm-j$.
\end{lemma}

\begin{definition}
Let $\boldsymbol\mu$ and $\boldsymbol\mu'$ be partitions of the same integer $r$. 
Then $\boldsymbol\mu$ is said to majorize $\boldsymbol\mu'$, denoted $\boldsymbol\mu\unrhd\boldsymbol\mu'$, if for all $j\ge1$, 
\[
\sum_{i=1}^j \boldsymbol\mu_i \ \ge \ \sum_{i=1}^j \boldsymbol\mu{'}_i
\] 
\end{definition}

\begin{lemma} \label{lem:majorization}
Let $\boldsymbol\mu\unrhd\boldsymbol\mu'$. Then for any distribution $\eigv$, $\Mapeigv{\boldsymbol\mu} \geq \Mapeigv{\boldsymbol\mu'}$.
\end{lemma}
Noting that $\flnpwr{n}{2\renprm}<(\flnpwr{n}{\renprm})^2$, we obtain 
\begin{align}
\Var\Paren{\proh{(\renprm)}{\lambda}} 
<&\sum_{\mu\in\cS}C_{\mu}\cdot \flnpwr{n}{\absv{\mu}} \Mapeigv{\mu}\nonumber\\
\stackrel{\text{(a)}} \le& c_{\alpha}\sum_{j=0}^{\renprm-1} \sum_{\mu \in \cS_j} \flnpwr{n}{\absv{\mu}} \Mapeigv{\mu}\nonumber\\
\stackrel{\text{(b)}} \le & c_{\alpha}\sum_{j=0}^{\renprm-1} \sum_{\mu \in \cS_j} \flnpwr{n}{\absv{\mu}} \Mapeigv{\renprm+j-\ell(\mu)+1}\nonumber\\
\stackrel{\text{(c)}} \le& c_{\alpha}\sum_{j=0}^{\renprm-1} \sum_{\mu \in \cS_j} \flnpwr{n}{\renprm+j} \Mapeigv{2j+1}\nonumber\\
\le& c_{\alpha}n^{\renprm}\sum_{j=0}^{\renprm-1} \sum_{\mu \in \cS_j} {n}^{j} \Mapeigv{2j+1}\nonumber\\
\stackrel{\text{(d)}} \le& c_{\alpha}n^{\renprm} |\cS| \max\{1,{n}^{\renprm-1} \Mapeigv{2\renprm-1}\},
\end{align} 
where (a) follows from the fact that $C_{\mu} \leq (\renprm !)^2 := c_{\alpha}$, (b) follows from Lemma~\ref{lem:majorization} and the fact that $[(\alpha+j-l(\mu)+1) \cup 1^{l(\mu)-1}] \unrhd \mu$, (c) follows from Lemma~\ref{lem:lengthpartition} and the fact that $\Mapeigv{r}$ is a non-increasing function in $r$ for fixed $\eigv$, and (d) follows from the fact that for $j\in\{1\upto\renprm-2\}$,\begin{align*}
{n}^{j-1} \Mapeigv{2j-1} \leq {n}^{j} \Mapeigv{2j+1} \Rightarrow {n}^{j} \Mapeigv{2j+1}  \le {n}^{j+1} \Mapeigv{2j+3}.
\end{align*}
Note that the above implication follows from Lemma~\ref{lem:majorization} applied to the partitions $(2j+1,2j+1)$ and $(2j-1,2j+3)$:
\begin{align*}
({n}^{j} \Mapeigv{2j+1})^2 \le {n}^{j-1} \Mapeigv{2j-1}\cdot{n}^{j+1} \Mapeigv{2j+3}.
\end{align*}
Finally, note that $|\cS|$ depends only on $\renprm$. Hence, the lemma follows by setting $C_{\renprm} = c_{\alpha} |\cS|$.

%% file: converse-integral.tex
\subsection{Converse}

Notice that there are two terms in the copy complexity in Theorem~\ref{thm:alpha-int}. The first term is $\dims^{1-1/\renprm}/\eps^2$. \cite{AcharyaOST17} showed that even in the classical setting, a lower bound of $\Omega(\dims^{1-1/\renprm}/\eps^2)$ holds. Invoking Lemma~\ref{lem:class-quant} gives the first term.

We use the classical Le Cam's method to prove the lower bound. We define a hypothesis testing problem below.
 
\paragraph{Two Point Testing.} Given density matrices $\rho$ and $\sigma$ with spectrums $\eigv$ and $\eignuv$, respectively. Let $\ns$ be given.
\begin{itemize}
\item 
Let $X$ be a uniform random variable over $\{0,1\}$. 
\item
If $X=0$, generate a Young tableau $\lamb\sim\swdist{\eigv}$.
\item
If $X=1$, generate a Young tableau $\lamb\sim\swdist{\eignuv}$. 
\item
Given $\lamb$, predict $X$ with $\hat X$.  
\end{itemize}
Let $P_e = \min_{\hat X} \probof{\hat X \ne X}$. 
From basic hypothesis testing results, we can deduce that 
\[
P_e = \frac12-\frac12\dtv{\swdist{\eigv}}{\swdist{\eignuv}}.
\]

We construct two spectrums $\eigv$ and $\eignuv$, such that $\rent{\renprm}{\eigv}- \rent{\renprm}{\eignuv}=\Theta(\eps)$, and  
\[
\frac12\dtv{\swdist{\eigv}}{\swdist{\eignuv}}<0.05,
\]
unless $\ns=\Omega(\dims^{2-2/\renprm}/\eps^{2/\renprm})$. This will prove that unless $\ns$ is large enough, there is no classifier that can test between the spectrums $\eigv$ and $\eignuv$ with probability greater than $2/3$, implying our lower bound.

Note that the second term in the complexity expression of Theorem~\ref{thm:alpha-int} dominates when $\eps>1/\sqrt{\dims}$. We henceforth assume in the remainder of this section that $\eps>1/\sqrt{\dims}$. 

Consider the following two spectrums:
\begin{align}
\eigv &= \Paren{\frac{1+(\eps\dims)^{1/\renprm}}{\dims}, \frac{1-\frac{(\eps\dims)^{1/\renprm}}{\dims-1}}{\dims},\ldots,\frac{1-\frac{(\eps\dims)^{1/\renprm}}{\dims-1}}{\dims}},\label{dist-lb-int} \\
 \eignuv &=\Paren{\frac{1}{\dims},\ldots,\frac{1}{\dims}}.\label{dist-unif}
\end{align}
Note that for any $d > 2$, assuming that\footnote{If $\eps \geq \log d$, then $\hat{S}_{\renprm}=0$ is a valid estimate and the problem becomes trivial.}   $\eps < \log d$, we have
\begin{equation}
\label{eq:depscond}
(\eps \dims)^{1/\alpha} < d-1.
\end{equation}
Thus $\eigv$ is a valid distribution.
$\eignuv$ corresponds to the maximally-mixed state.

\begin{lemma} 
\label{clm:difference}
Suppose $\eps>1/\sqrt{d}$ and $\dims>(3 \renprm)^{\frac{2\renprm}{\renprm-1}} $. Then
\[
| \rent{\renprm}{\eignuv}-\rent{\renprm}{\eigv}  | \geq \frac{1}{\renprm -1 } \log \left(1 + \frac{2\eps}{3}\right). 
\]
\end{lemma} 
\begin{proof}
Computing the moments of $\eigv$, we have
\begin{align}
\Map{\renprm}{\eigv} &= \frac1{\dims^{\renprm}}\Paren{\Paren{1+(\eps\dims)^{1/\renprm}}^{\renprm}+(\dims-1)\Paren{1-\frac{(\eps\dims)^{1/\renprm}}{\dims-1}}^{\renprm}}.\nonumber
\end{align}
For $\renprm\ge1$  and $x\ge0$, note that $(1+x)^\renprm>1+x^{\renprm}$, and, if $x \le 1$, $(1-x)^\renprm>1-\renprm x$. Using these two inequalities above with $x=(\eps\dims)^{1/\renprm}$ in the first term, and with $x = (\eps\dims)^{1/\renprm}/(\dims-1)$ in the second term (and using (\ref{eq:depscond})), we obtain 
\begin{align}
\Map{\renprm}{\eigv} &= \frac1{\dims^{\renprm}}\Paren{\Paren{1+(\eps\dims)^{1/\renprm}}^{\renprm}+(\dims-1)\cdot\Paren{1-\frac{(\eps\dims)^{1/\renprm}}{\dims-1}}^{\renprm}}\nonumber\\
&\ge \frac1{\dims^{\renprm}}\Paren{1+\eps\dims+(\dims-1)\cdot\Paren{1-\frac{\renprm(\eps\dims)^{1/\renprm}}{\dims-1}}}\nonumber\\
&\ge \frac1{\dims^{\renprm}}\Paren{\dims+\eps\dims-{{\renprm(\eps\dims)^{1/\renprm}}}}\nonumber\\
& \ge\frac{d}{\dims^{\renprm}}\Paren{1+\frac23\eps}\nonumber\\
& = {\Map{\renprm}{\eignuv}}\cdot\Paren{1+\frac23\eps},\nonumber
\end{align}
whenever $d>\frac{(3 \renprm)^{\frac{\renprm}{\renprm-1}}}{\eps}$, which is implied by the conditions
$\eps > 1/\sqrt{d}$ and $\dims > (3 \renprm)^{\frac{2\renprm}{\renprm-1}}$.
\end{proof}

\begin{lemma}
\label{thm:all-dim} 
Any algorithm that can test between $\eigv$ and $\eignuv$ with probability at least $2/3$ requires at least $\Omega\Paren{\frac{\dims^{2-2/\renprm}}{\eps^{2/\renprm}}}$ copies.
\end{lemma}

\begin{proof}
We prove that $\dtv{\swdist{\eigv}}{\swdist{\eignuv}}<0.05$. Bounding the total variation distance is hard to handle, and therefore other distance measures are used to bound the total variation distance. By Lemma~\ref{lem:distance-bounds}, we know that
\[
2\dtv{\swdist{\eigv}}{\swdist{\eignuv}}^2\le \chisq{\swdist{\eigv}}{\swdist{\eignuv}}.
\]
The objective is to bound the $\chi^2$ distance between the SW distributions for the two states with $\ns$ copies. We use the following formula, derived in~\cite[Corollary~6.2.4]{Wright16}. The result in this form was obtained from related results on Schur functions~\cite{OkounkovO98}. 
\begin{lemma}
Let $x_1,\ldots,x_\dims$ be such that $\sum x_i = 0$, and $x_i\ge-1$. Let $\eigv$ be the spectrum with $\eigvi{i} =(1+x_i)/\dims$, and $\eignuv$ be the spectrum of the maximally mixed state, namely $\eignuv_i = 1/\dims$. Then,
\[
\chisq{\swdist{\eigv}}{\swdist{\eignuv}}=\sum_{\mu:1\le\ell(\mu)\le\dims}\frac{\schurlx{\mu}{x}^2}{\dims^{\overline{\mu}}\dims^{|\mu|}}\ns^{\underline{|\mu|}},
\] 
where \newest{for a partition $\mu$, $\dims^{\overline{\mu}}$ is defined below.} 
\end{lemma}

\newest{\begin{definition}\label{def:falling}
Let $\mu$ be a partition. Index each box in the Young tableaux for $\mu$ with an entry $(i,j)$, where $i$ the row number and $j$ is the column number of the box. For each box $\square$ in the tableaux, let $c(\square) = j-i$ be the content of $\square$. Then for a real number $z\in\RR$,
\[
z^{\overline{\mu}} = \prod_{\square} (z+c(\square)).
\]
\end{definition}}

\newest{We will use the following bound on these falling powers of partitions to prove our lower bound. 
\begin{lemma}\label{lem:falling}
Let $\mu$ be a partition such that $\ell(\mu)\le d$, where $\ell(\mu)$ is the number of non-zero entries of $\mu$ (which is also the number of non-empty rows in the Young tableaux. Then
\[
\dims^{\overline{\mu}} \ge \Paren{\frac{\dims}{e}}^{\absv{\mu}}.
\]
\end{lemma}}
This result is proved in Appendix~\ref{app:lemma-falling}, and we now prove our result using this lemma. 

The distribution $\eignuv$ corresponds to the spectrum defined in~\eqref{dist-unif}, and we choose the $x_i$'s to make the spectrum $\eigv$ equal to~\eqref{dist-lb-int}. In particular, let $x_1 = (\eps\dims)^{1/\renprm}$, and $x_i = - \frac{(\eps\dims)^{1/\renprm}}{\dims-1}$ for $i=2,\ldots,\dims$. Let $y_1 = 1$, and $y_i = -1/(\dims-1)$ for $i = 2,\ldots, \dims$. Then, 

\[
x^{\dims}_1 = (\eps\dims)^{1/\renprm}\cdot \Paren{1, \frac{-1}{\dims-1}, \ldots, \frac{-1}{\dims-1}} = (\eps\dims)^{1/\renprm} \cdot y^{\dims}_1.
\]

Recall that the Schur polynomial $\schurlx{\mu}{x_1^d}$ is a homogeneous symmetric polynomial of degree $\absv{\mu}$. This implies, 
\begin{align}
	\schurlx{\mu}{x_1^{\dims}} = \Paren{\eps\dims}^{\frac{\absv{\mu}}{\renprm}}\schurlx{\mu}{y^{\dims}_1}\label{eqn:schur-xy}.
\end{align}
Let $\ydp$ be the vector of absolute values of $\yd$, namely
\[
\ydp =\Paren{1, \frac1{\dims-1}, \ldots, \frac1{\dims-1}}.
\]
Then, $\absv{\schurlx{\mu}{\yd}}\le\absv{\schurlx{\mu}{\ydp}}$. Using the fact that $\dims^{\overline{\mu}}\ge (\dims/e)^{|\mu|}$, and $\ns^{\underline{m}} \le \ns^{m}$,
\begin{align}
\chisq{\swdist{\eigv}}{\swdist{\eignuv}}
=&\sum_{{\mu:1\le\ell(\mu)\le\dims}}\frac{\schurlx{\mu}{x}^2}{\dims^{\overline{\absv\mu}}\dims^{|\mu|}}\ns^{\underline{|\mu|}} \nonumber\\ 
= & \sum_{ \substack{ {\mu:1\le\ell(\mu)\le\dims} \\ \newest{ 1 < \absv{\mu} \leq n } } }\frac{\schurlx{\mu}{x}^2}{\dims^{\overline{\absv\mu}}\dims^{|\mu|}}\ns^{\underline{|\mu|}} \nonumber\\ 
\le & \sum_{ \substack{ {\mu:1\le\ell(\mu)\le\dims} \\ \newest{ 1 < \absv{\mu} \leq n } }  }\schurlx{\mu}{\ydp}^2\cdot\Paren{\frac{\ns(\eps\dims)^{2/\renprm}}{\newest{(\dims/e)}\cdot\dims}}^{\absv{\mu}} \nonumber\\
= & \sum_{ \substack{ {\mu:1\le\ell(\mu)\le\dims} \\ \newest{ 1 < \absv{\mu} \leq n } }  }\schurlx{\mu}{\ydp}^2\cdot\Paren{\frac{\newest{e}\cdot\ns\eps^{2/\renprm}}{\dims^{2-2/\renprm}}}^{\absv{\mu}}\nonumber\\
\le & \sum_{m=2}^{\ns}\Paren{\Paren{\frac{\newest{e}\ns\eps^{2/\renprm}}{\dims^{2-2/\renprm}}}^{m}\cdot\Paren{\sum_{\mu:|\mu|=m}\schurlx{\mu}{\ydp}^2}},\nonumber
\end{align}
\newest{where the second equality follows from the fact that $\ns^{\underline{|\mu|}} = 0$ for $\absv{\mu} > n$.} Let $p(m)$ denote the partition number of $m$, the number of unordered partitions of $m$. Bounds on the growth of partition numbers are well established~\cite{HardyR18}. We only require the following loose upper bound that holds for all $m$
\[
p(m)<e^{3\sqrt{m}}. 
\]
This gives
\begin{align}
\chisq{\swdist{\eigv}}{\swdist{\eignuv}}\le \sum_{m=2}^{\ns}\Paren{e^{3\sqrt{m}}\Paren{\frac{\ns\eps^{2/\renprm}}{\dims^{2-2/\renprm}}}^{m}\cdot \max_{\mu:|\mu|=m}\schurlx{\mu}{\ydp}^2}.\label{eqn:chi-main}
\end{align}

The entries of $\ydp$ have the following structure. The first entry is 1, and all other entries are $1/(\dims-1)$. This allows us to use the ``branching rule'' of Schur polynomials. The general form can be found in~\cite[Eq 5.10]{Macdonald98}. A special case appears in the following form in~\cite[Eq. 1.4]{LascouxW11b}.
\begin{lemma} The Schur polynomial $\schurlx{\mu}{z^d_1}$ can be decomposed as:
\begin{align}
\schurlx{\mu}{z^d_1} = \sum_{\lambda\prec \mu} (z_1)^{\absv{\mu} - \absv{\lambda}} \schurlx{\lambda}{z_2^{\dims}},\label{eqn:branching-rule}
\end{align}
where the summation is over all partitions $\lambda$ such that and $\mu_1\ge \lambda_1\ge\mu_2\ge\lambda_2\ge\mu_3\ge\ldots$,
\end{lemma}

Applying this with $z^\dims=\ydp$, 
\begin{align}
\schurlx{\mu}{\ydp} 
= \sum_{\lambda\prec \mu}\Paren{\frac{1}{\dims-1}}^{|\lambda|}\schurlx{\lambda}{1^{\dims-1}}.
\end{align}
From (\ref{eq:schurdef}), we see that 
$\schurlx{\lambda}{1^{\dims-1}}$ is the number of semistandard Young tableaux with shape $\lambda$ and entries from $[\dims-1]$. 
We can trivially bound $\schurlx{\lambda}{1^{\dims}}\le(\dims-1)^{\absv{\lambda}}$, the total number of ways of filling the Young tableaux with entries from $[d]$, \emph{without any regard to ordering}. 

We need one final definition.
\begin{definition}
For a partition $\mu$, let $\mathrm{prec}(\mu)$ be the number of partitions $\lambda$ such that $\lambda\prec\mu$. 
\end{definition}

\begin{lemma}
\label{lem:prec}
\begin{align*}
\mathrm{prec}(\mu) = \prod_{i=1}^{\infty}\Paren{\mu_i-\mu_{i+1}+1}< m^{\sqrt {2 m}}.
\end{align*}
\end{lemma}
\begin{proof}
The equality is due to a simple counting argument. For the inequality,
let $\mltu_{i_1}>\mltu_{i_2}>\ldots>\mltu_{i_k}\ge1$ be the distinct elements in $\mltu$. If $k=1$, the inequality is easy to show, so assume that $k > 1$. 
Then, $k(k+1)/2\le \mltu_{i_1}+\ldots+\mltu_{i_k} \leq m$, implying that $k <\sqrt {2 m}$. Moreover, $\mltu_{i_1}-\mltu_{i_k}\le \absv{\mltu}-1=m-1$ since $\mltu_{i_k} > 1$.
\begin{align}
\text{prec}(\mu)\le \prod_{j=1}^{k} \Paren{1+\mltu_{i_j}-\mltu_{i_{j+1}}}
\le m^k<m^{\sqrt {2m}}.\nonumber
\end{align}
\end{proof}
Therefore, 
\begin{align}
\schurlx{\mu}{y_+^d} 
= \sum_{\lambda\prec \mu}\Paren{\frac{1}{\dims-1}}^{|\lambda|}\schurlx{\lambda}{1^{\dims-1}} \le \sum_{\lambda\prec \mu}\Paren{\frac{1}{\dims-1}}^{|\lambda|}\Paren{\dims-1}^{|\lambda|}= \text{prec}(\mu)\le \absv{\mu}^{\sqrt {2 \absv{\mu} }}.
\label{bound-schur-sp}
\end{align}


Plugging~\eqref{bound-schur-sp} in~\eqref{eqn:chi-main}, 
\begin{align}
\chisq{\swdist{\eigv}}{\swdist{\eignuv}}&\le \sum_{m=2}^{\ns}\Paren{e^{3\sqrt m}\Paren{\frac{\newest{e}\ns\eps^{2/\renprm}}{\dims^{2-2/\renprm}}}^{m}\cdot \max_{\mu:|\mu|=m}\schurlx{\mu}{\ydp}^2}\nonumber\\
&\le \sum_{m=2}^{\ns}\Paren{e^{3\sqrt m}\Paren{\frac{\newest{e}\ns\eps^{2/\renprm}}{\dims^{2-2/\renprm}}}^{m}\cdot m^{2\sqrt{2m}}}\nonumber\\
&\le \sum_{m=2}^{\ns}\Paren{(em)^{3\sqrt m}\Paren{\frac{\newest{e}\ns\eps^{2/\renprm}}{\dims^{2-2/\renprm}}}^{m}}.\nonumber
\end{align}
Finally note that $m^{\sqrt m} <2\cdot 2^{m}$ for all $m>1$. Therefore, 
\begin{align}
\chisq{\swdist{\eigv}}{\swdist{\eignuv}}
\le \sum_{m=2}^{\ns}8\Paren{\frac{(2e)^4\ns\eps^{2/\renprm}}{\dims^{2-2/\renprm}}}^{m}.\nonumber
\end{align}
Therefore, unless $n\ge\Omega\Paren{\frac{\dims^{2-\frac2{\renprm}}}{\eps^{\frac{2}{\renprm}}}}$, the $\chi^2$ distance is small, proving the result.
\end{proof}

%% file: von-neumann.tex
\section{von Neumann Entropy}

\subsection{Empirical Entropy Upper Bound}

Analogous to the classical setting, the empirical distribution is
\[
\eighati i \ed \frac{\lambi{i}}{\ns}.
\]
The empirical estimate of $\entmst$ is
\[
\empents{\lamb} \ed \sum_{i=1}^{\dims} \frac{\lambi{i}}{\ns}\log\frac{\ns}{\lambi{i}} = \sum_{i=1}^{\dims} \eighati{i}\log\frac{1}{\eighati{i}}.
\]

We prove the following bound on the mean squared error of this estimator.
\begin{theorem}
\label{thm:von-neumann-empirical}
The empirical entropy estimate satisfies:
\[
\expectation{\Paren{\empents{\lamb}-\entmst}^2} \le O\Paren{\frac{\dims^4}{\ns^2}+\frac{\dims^2}{\ns}+\frac{\log^2\ns}{\ns}}.
\]
\end{theorem}
An immediate corollary is the following sample complexity bound. 
\begin{corollary}
\[
\copycmph= O\Paren{\frac{\dims^2}{\eps^2}+\frac{\log^2(1/\eps)}{\eps^2}}.
\]
\end{corollary}
\begin{proof}
By Markov's Inequality on Theorem~\ref{thm:von-neumann-empirical}, there is a constant $C$ such that with probability at least $0.9$, 
\[
\absv{\Paren{\empent-\entmst}}< C \sqrt{{\frac{\dims^4}{\ns^2}+\frac{\dims^2}{\ns}+\frac{\log^2\ns}{\ns}}}
< C \Paren{{{\frac{\dims^2}{\ns}}}+{\frac{\dims}{\sqrt{\ns}}}+{\frac{\log\ns}{\sqrt\ns}}}.
\]
Bounding each term to at most $\eps/3C$ gives the sample complexity bound. 
\end{proof}

\begin{proof}[Proof of Theorem~\ref{thm:von-neumann-empirical}]
	
\end{proof}

For an estimator $\hat X$ of a parameter $x$, the mean-squared error can be decomposed as
\begin{align}
\expectation{(x-\hat X)^2} = \expectation{\Paren{x-\expectation{\hat X}}^2}+\expectation{\Paren{\hat X-\expectation{\hat X}}^2},\nonumber
\end{align}
where the first term is the squared bias, and the second term is the variance. In particular, 
\begin{align}
\expectation{\Paren{\empents{\lamb}-\entmst}^2} = \Paren{\entmst-\expectation{\empents{\lamb}}}^2+\Var\Paren{\empents{\lamb}}.\label{eqn:bias-var}
\end{align}

The theorem follows by plugging the following two bounds on the bias and variance into~\eqref{eqn:bias-var}. 
\begin{lemma}
\label{lem:bias-von-neumann}
\[
\absv{\entmst -\expectation{ \empents{\lamb}}} \le \frac{\dims^2}{n} + 9\frac{\dims}{\sqrt \ns}.
\]
\end{lemma}

\begin{lemma}
\label{variance-von-neumann}
\[
\Var\Paren{\empents{\lamb}} = O\Paren{\frac{\log^2 \ns}{\ns}}.
\]
\end{lemma}

\subsubsection{Bounding the Bias (Proof of Lemma~\ref{lem:bias-von-neumann})}

The bias of the empirical estimate can be bounded as:
\begin{align}
\absv{ \entmst-\expectation{\empents{\lamb} } }
= & \absv{\expectation{\sum_{i=1}^{\dims} \Paren{ \eigvi{i}\log \frac1{\eigvi{i}} - \eighati{i}\log\frac{1}{\eighati{i}}}}}\nonumber\\
=& \absv{\expectation{\sum_{i=1}^{\dims} \Paren{ \eigvi{i}\log\frac{1}{\eigvi{i}} - \eighati{i}\log \frac1{\eigvi{i}}
+ \eighati{i}\log\frac{1}{\eigvi{i}} - \eighati{i}\log \frac1{\eighati{i}}}}}\nonumber\\
\le & \absv{\expectation{\sum_{i=1}^{\dims} (\eigvi{i} - \eighati{i})\log\frac{1}{\eigvi{i}} }}
+ \absv{\expectation{\sum_{i=1}^{\dims}\Paren{ \eighati{i}\log\frac{\eighati{i}}{\eigvi{i}}}}}\nonumber\\
\le & \absv{\sum_{i=1}^{\dims} (\eigvi{i} - \expectation{\eighati{i}})\log\frac{1}{\eigvi{i}} }
+ \expectation{\sum_{i=1}^{\dims}\frac{(\eighati{i}-\eigvi{i})^2}{\eigvi{i}}}.\label{eqn:vonneumann-ub}
\end{align}

The second term is the expected $\chi^2$-distance of the empirical distance and the underlying distribution. Theorem 4.7 of~\cite{OW17} states that
\[
\expectation{\sum_{i=1}^{\dims}\frac{(\eighati{i}-\eigvi{i})^2}{\eigvi{i}}}\le \frac{\dims^2}{\ns},
\]
which bounds the second term of~\eqref{eqn:vonneumann-ub}. We now bound the first term. We again use the following result from~\cite{OW17} that bounds the expected value of $\eighati{i}$ around $\eigvi{i}$. 
\begin{lemma}
[Theorem 1.4 of~\cite{OW17}]
\[
\absv{{\eigvi i-\expectation{\Paren{\eighati{i}}}}}\le 2\sqrt{\frac{\min\left\{1, \eigvi i\dims\right\}}{\ns}}.
\]
\label{lem:conc-emp}
\end{lemma}

Let $c_1, \ldots, c_\dims$ be the constants such that
$\eigvi{i}-\expectation{\eighati{i}} = c_i \sqrt{\frac{\dims\eigvi{i}}{\ns}},$
then, by Lemma~\ref{lem:conc-emp}, $\absv{c_i}\le 2$. 
Since $\sum_{i=1}^{\dims}\eigvi{i}={\sum_{i=1}^{\dims}\eighati{i}}=1$, 
$$\sqrt{\frac \dims\ns}\Paren{\sum_{i=1}^{\dims} c_i\sqrt{\eigvi{i}}} = \sum_{i=1}^{\dims}\Paren{\eigvi{i}-\expectation{\eighati{i}}} =0,$$
implying that $\sum_{i=1}^{\dims} c_i\sqrt{\eigvi{i}}=0$. 
Therefore, 
\begin{align}
	\sum_{i=1}^{\dims} (\eigvi{i} - \expectation{\eighati{i}})\log\frac{1}{\eigvi{i}}
=  \sqrt{\frac{\dims}{\ns}}\cdot \Paren{\sum_{i=1}^{\dims} c_i\sqrt \eigvi{i}\log\frac{1}{\eigvi{i}}}.\label{eqn:opt-problem}
\end{align}
Since $\sqrt{\frac{\dims}{\ns}}$ is a constant, to bound the first term of~\eqref{eqn:vonneumann-ub} it will suffice to upper bound the following maximization problem. 
\begin{align}
{\bf P1:}\qquad&\text{maximize}\ \ \absv{\sum_{i=1}^{\dims} c_i\sqrt \eigvi{i}\log\frac{1}{\eigvi{i}}}\nonumber\\
& \text{subject to } \absv{c_i}\le 2, \text{ and }\ \sum_{i=1}^{\dims} c_i\sqrt{\eigvi{i}}=0.\nonumber
\end{align} 
By the triangle inequality, 
\begin{align}
\absv{\sum_{i=1}^{\dims} c_i\sqrt \eigvi{i}\log\frac{1}{\eigvi{i}}}
\le \absv{\sum_{i=1}^{\dims} c_i\sqrt \eigvi{i}\log\frac{1}{c_i^2\eigvi{i}}} + \absv{\sum_{i=1}^{\dims} c_i\sqrt \eigvi{i}\log{c_i^2}}\label{eqn:two-parts-bias}
\end{align}
We bound the terms individually. We first consider the second term. Since $|c_i|\le 2$, the largest value of $|c_i \log c_i^2|$ is $2\log 4$. Therefore, 
\[
\absv{\sum_{i=1}^{\dims} c_i\sqrt \eigvi{i}\log{c_i^2}}
\le 2\log 4\cdot \Paren{\sum_{i=1}^{\dims} \sqrt \eigvi{i}} \le (2\log 4) \cdot\sqrt \dims,
\]
where we use that $\sum_{i=1}^{\dims} \sqrt \eigvi{i}<\sqrt{\dims}$ by concavity of square root.

Let $x_i=c_i\sqrt \eigvi{i} $, then $\sum_i x_i =0$, and since $\sum \eigvi{i} =1$, $\sum_i x_i^2 \le 4$. Therefore, to bound the first term of~\eqref{eqn:two-parts-bias}, it will suffice to solve {\bf P2} below. 
\begin{align}
{\bf P2:}\qquad &  \text{maximize}\ \ \sum_{i=1}^\dims x_i \log \frac1{x_i^2}\\
&  \text{subject to}\ \  \sum_{i=1}^\dims x_i = 0, \text{ and }\ \sum_{i=1}^\dims x_i^2 \le 4.
\end{align}

\noindent We show in Appendix~\ref{sec:app-bias} that 
\begin{lemma}
\label{lem:opt-problem}
The maximum value of the optimization problem {\bf P2} is at most $\frac{16}e\sqrt \dims.$
\end{lemma}

Plugging this in~\eqref{eqn:two-parts-bias}, the maximum of {\bf P1} is at most $(16/e+2\log 4)\sqrt \dims$. Therefore,  
\[
\absv{\sum_{i=1}^{\dims} (\eigvi{i} - \expectation{\eighati{i}})\log\frac{1}{\eigvi{i}} } \le \Paren{\frac {16}e+2\log 4}\frac{\dims}{\sqrt\ns} \le \frac{9\dims}{\sqrt{\ns}}.
\]

Plugging this in turn into~\eqref{eqn:vonneumann-ub} yields
\[
\absv{ \entmst-\expectation{\empents{\lamb}}}\le \frac{\dims^2}{\ns}+\frac{9\dims}{\sqrt\ns},
\]
thus bounding the bias.

\subsubsection{Proof of Lemma~\ref{variance-von-neumann}.}
We will use the bounded difference variance bound (Lemma~\ref{lem:bdd-diff-var}). In particular, we consider the non-decreasing subsequence interpretation of weak Schur sampling. Let $X^\ns\in[\dims]^\ns$,  and let $\lamb$ be the shape of its young tableaux through the RSK correspondence. Let $\lamb'$ be the shape of the Young tableaux corresponding to a sequence with Hamming distance at most one from $X^n$. Let $\ent{\lamb}$, and $\ent{\lamb'}$ denote their respective empirical von Neumann entropy. The next lemma states that changing one of the $n$ symbols has \emph{small effect} on the empirical entropy. 

\begin{lemma} Let $\lamb$, and $\lamb'$ be two Young tableaux shapes obtained from the LIS of two length-$\ns$ samples that differ in at most one symbol. 
If $n \ge 27$, then
\[
\absv{\empents{\lamb}-\empents{\lamb'}} \le \frac{15\log \ns}{\ns}.
\]
\label{lem:lipschitz}
\end{lemma}
\noindent This lemma is proved in Appendix~\ref{sec:app-lipschitz-entropy}.

We invoke the bounded difference inequality (Lemma~\ref{lem:bdd-diff-var}) along with Lemma~\ref{lem:lipschitz}. The empirical entropy estimate changes by at most $15\log \ns/\ns$ when one symbol is changed. Therefore, the variance is at most
\[
\Var\Paren{\empents{\lamb}} \le\frac14\ns\cdot \Paren{\frac{15\log \ns}{\ns}}^2\le \frac{225\log^2 \ns}{4\ns}.
\]

%% file: non-integral-large-alpha.tex
\section{Non Integral $\alpha$}
\label{sec:non-int}
\subsection{$\alpha>1$}

We prove the following sample complexity bound for estimating $\rentprmmst$ for $\renprm>1$. 
\begin{theorem}
\label{thm:non-int-upper-large-alpha}
For $\renprm>1$, the empirical estimator of $\rentprmmst$ outputs a $\pm \eps$ estimate with  $O\Paren{{\frac{\dims^2}{\eps^2}}}$ copies of $\mst$ with probability at least 0.9.
\end{theorem}
\begin{proof}
Recall that $\ns = \sum \lambi{i}$. Define 
\[
\Malamb{\renprm} \ed \sum_{i=1}^{\dims}\Paren{\frac{\lambi{i}}{\ns}}^{\renprm}.
\]
We show that when $\ns$ is large enough, $\Malamb{\renprm}$ is within a small multiplicative factor of $\Mapeigv{\renprm}$. The following result shows each term $(\lambi{i}/\ns)^\renprm$ concentrates around $\eigvi{i}^\renprm$. 
\begin{lemma}
Let $\beta>1$, and further suppose that the sorted probabilities are $\eigvi{i}$. Then there is a constant $C_{\beta}$ such that 
\[
\expectation{\absv{\lambi{i}^{\beta}-(\eigvi{i}\ns)^\beta}}< C_{\beta}\cdot \Paren{n^{\beta/2} + \sqrt n(\eigvi{i}\ns)^{\beta-1}}.
\]
\label{lem:diff-moments}
\end{lemma}
\noindent This lemma is proved in Appendix~\ref{sec:large-alpha}.

Then,
\begin{align}
\expectation{\absv{\Malamb{\renprm}-\Mapeigv{\renprm}}} 
& = \expectation{\absv{\sum_{i=1}^{\dims}\Paren{\Paren{\frac{\lambi{i}}{\ns}}^{\renprm}-\eigvi{i}^\renprm}}}\nonumber\\
&\le \frac1{\ns^{\renprm}}\sum_{i=1}^{\dims}\expectation{\absv{\lambi{i}^{\renprm}-(\eigvi{i}\ns)^\renprm}}\nonumber\\
&\le \frac{C_\renprm}{\ns^{\renprm}}\sum_{i=1}^{\dims} \Paren{\ns^{\renprm/2} + \sqrt \ns(\eigi{i}\ns)^{\renprm-1}} \label{eqn:moment-bound-large-alpha}\\
& = C_\renprm\Paren{\frac{\dims}{\ns^{\renprm/2}} + \frac1{\sqrt \ns}\sum_{i=1}^{\dims}\eigvi{i}^{\renprm-1}} \nonumber\\
& = C_\renprm\Paren{\frac{\dims}{\ns^{\renprm/2}} + \frac{\Mapeigv{\renprm-1}}{\sqrt \ns}},\label{eqn:non-int-bias}
\end{align}
where~\eqref{eqn:moment-bound-large-alpha} uses Lemma~\ref{lem:diff-moments}.

By Lemma~\ref{lem:bnd_moments}, for $\renprm>1$, $\Mapeigv{\renprm}\ge \dims^{1-\renprm}$, and 
$\Mapeigv{\renprm-1}\le \dims\Mapeigv{\renprm}$. Substituting in~\eqref{eqn:non-int-bias},
\begin{align}
\expectation{\absv{\Malamb{\renprm}-\Mapeigv{\renprm}}} 
& \le C_\renprm \Paren{\frac{\dims}{\ns^{\renprm/2}} + \frac{\Mapeigv{\renprm-1}}{\sqrt \ns}}\nonumber\\
& \le C_\renprm \Paren{\frac{\dims^{\renprm}\Mapeigv{\renprm}}{\ns^{\renprm/2}} + \frac{\dims\Mapeigv{\renprm}}{\sqrt \ns}}\nonumber\\
& \le C_\renprm \Paren{\frac{\dims^{\renprm}}{\ns^{\renprm/2}} + \frac{\dims}{\sqrt \ns}}\Mapeigv{\renprm}.\nonumber
\end{align}
By Markov's Inequality, 
\begin{align}
\probof{\absv{\Malamb{\renprm}-\Mapeigv{\renprm}}>\eps \Mapeigv{\renprm}}
\le \frac{\expectation{\absv{\Malamb{\renprm}-\Mapeigv{\renprm}}}}{\eps\Mapeigv{\renprm}} \le \frac{C_\renprm}{\eps}\Paren{\frac{\dims^{\renprm}}{\ns^{\renprm/2}} + \frac{\dims}{\sqrt \ns}}.\nonumber
\end{align}
Therefore, when $\ns > C \dims^2\Paren{\frac1{\eps^2}+\frac1{\eps^{2/\renprm}}}$, the result follows. Since $\renprm>1$, the first term dominates. 
\end{proof}

%% file: non-integral-small-alpha.tex
\subsection{$\renprm<1$}

In this section, we will prove the following:
\begin{theorem}
\label{thm:non-int-upper-small-alpha}
The empirical estimator of $\rentprmmst$ outputs a $\pm \eps$ estimate with  $O\Paren{\Paren{\frac{\dims}{\eps}}^{2/\renprm}}$ copies.
\end{theorem}

Similar to the case of large $\renprm$, we need the following result, which is proved in Appendix~\ref{sec:app-small-beta}.
\begin{lemma}
\label{lem:diff-moment-small-alpha}
Let $\beta<1$ and suppose that the sorted probabilities are $\eigvi{i}$. Then there is a constant $C_\beta$ such that 
\[
\expectation{\absv{\lambi{i}^{\beta}-(\eigi{i}\ns)^\beta}}< C_{\beta}\cdot \ns^{\beta/2}.
\]
\end{lemma}

We now prove the copy complexity bound assuming this result. 
\begin{proof}[Proof of Theorem~\ref{thm:non-int-upper-small-alpha}]
Recall that $\ns = \sum \lambi{i}$. Define, 
\[
\Malamb{\renprm} \ed \sum_{i=1}^{\dims}\Paren{\frac{\lambi{i}}{\ns}}^{\renprm}.
\]
Then by the triangle inequality, 
\begin{align}
\expectation{\absv{\Malamb{\renprm}-\Mapeigv{\renprm}}} 
&\le \frac1{\ns^{\renprm}}\sum_{i=1}^{\dims}\expectation{\absv{\lambi{i}^{\renprm}-(\eigvi{i}\ns)^\renprm}}\nonumber\\
&\le \frac{C_\renprm}{\ns^{\renprm}}\sum_{i=1}^{\dims} \ns^{\renprm/2} \label{eqn:smallb}\\
& = C_\renprm{\frac{\dims}{\ns^{\renprm/2}}},  \label{eqn:non-int-bias-small}
\end{align}
where~\eqref{eqn:smallb} follows from Lemma~\ref{lem:diff-moment-small-alpha}.
For $\renprm<1$, $\Mapeigv{\renprm}\ge1$. Substituting in~\eqref{eqn:non-int-bias-small},
\begin{align}
\expectation{\absv{\Malamb{\renprm}-\Mapeigv{\renprm}}} 
 \le C_\renprm{\frac{\dims}{\ns^{\renprm/2}}}  \le C_\renprm{\frac{\dims}{\ns^{\renprm/2}}} \Mapeigv{\renprm}. \nonumber
\end{align}
By Markov's Inequality, 
\begin{align}
\probof{\absv{\Malamb{\renprm}-\Mapeigv{\renprm}}>\eps \Mapeigv{\renprm}}
\le \frac{\expectation{\absv{\Malamb{\renprm}-\Mapeigv{\renprm}}}}{\eps\Mapeigv{\renprm}}
 \le \frac{C_\renprm}{\eps}\Paren{\frac{\dims}{\ns^{\renprm/2}}}.\nonumber
\end{align}
Therefore, when $\ns > C \Paren{\frac{\dims}{\eps}}^{2/\renprm}$, the result follows.
\end{proof}

%% file: converse-large-alpha-empirical.tex
\section{Lower bound on the performance of empirical entropy estimate}
\label{sec:lb-large}

\subsection{Stronger error probability bounds for performance of EYD}
The EYD algorithm for estimating the multiset of probability elements simply outputs the empirical probabilities of the Young tableaux, namely $\eigvi{i}=\lambi{i}/\ns$. The performance of the EYD algorithm has been well studied. We will consider the special case of the uniform distribution, and the performance metric of total variation. It is known~\cite{OW16} that the EYD algorithm using $O(\dims^2/\eps^2)$ samples from the uniform distribution satisfies with high probability,  
\[
\expectation{\dtv{\eigv}{u}}<\eps. 
\]
The best known lower bounds for the performance of the EYD algorithm is the following result of~\cite{ODonnellW15}. 
\begin{theorem}
\label{thm:od-bound}
There is a constant $\eps_0>0$, such that for $\eps\le\eps_0$, 
\[
\probof{\sum_{i=1}^{\dims}\absv{\frac{\lambi{i}}{\ns}-\frac1\dims}>\eps}>0.01
\]
unless $\ns = \Omega(\dims^2/\eps^2)$.
\end{theorem}

We will strengthen their error probability bound as follows. 
\begin{theorem}
There are constants $\eps_0>0$,  $c_1$, and $c_2$ such that when $\eps\le\eps_0$, and $\ns<c_1\dims^2/\eps^2$, and $\rho$ is maximally mixed,
\[
\probof{\sum_{i=1}^{\dims}\absv{\frac{\lambi{i}}{\ns}-\frac1\dims}>\eps}>1-\exp(-c_2\cdot d).
\]
\label{thm:emp-tv-unif}
\end{theorem}
\begin{proof}
	\noindent Let
\[
Z(\lamb) = \sum_{i=1}^{\dims} \absv{\frac{\lambi{i}}{\ns}-\frac1\dims}.
\]
By the LIS interpretation of the Young tableaux, we will show that if $\lamb$, and $\lamb'$ are two Young tableaux corresponding to sequences that differ at at most one position, then the difference of their total variation distances from the uniform distribution is small. In particular, 
\begin{lemma}
\[
|Z(\lamb)-Z(\lamb')| \le \frac{14}{\ns}.
\]
\label{lem:dtv-lipschitz}
\end{lemma}
\noindent This lemma is proved in Section~\ref{app:lipschitz}. 

To prove Theorem~\ref{thm:emp-tv-unif}, we first show that it holds for small $\ns$ (at most $O(\dims/\eps^2)$). This is proved in the following two lemmas and relies only on results about the empirical distribution in the classical setting. These are proved in Appendix~\ref{app:lipschitz}.
\begin{lemma}
\label{lem:l-one-error}
Let $\hat p$ be the empirical distribution from $\ns$ draws from the uniform distribution $u$ over $[\dims]$. There are constants $\eps_0$, and $c$, such that for any $\eps<\eps_0$, unless $\ns=\Omega(\dims/\eps^2)$, 
\[
\probof{\dtv{\hat p}{u}<\eps/2}<\exp(-c\cdot\dims).
\]
\end{lemma}

Using a coupling argument, we prove the following lemma. 
\begin{lemma}
\label{lem:emp-quant}
Unless $\ns=\Omega(\dims/\eps^2)$, 
\[
\probof{Z(\lamb)<\eps}\le \probof{\dtv{\hat p}{u}<\eps/2}<\exp(-c\cdot\dims),
\]
where $c$ is the same constant as Lemma~\ref{lem:l-one-error}.
\end{lemma}
We henceforth assume that $\ns>C\dims/\eps^2$ for some constant $C$. 
By Theorem~\ref{thm:od-bound} with appropriate normalization, we can claim that for $\eps<\eps_0$ (for some constant $\eps_0$), unless $\ns=\Omega(\dims^2/\eps^2)$, 
\begin{align}
\expectation{Z(\lamb)}>\eps.\label{eqn:expectation-z}
\end{align}
Therefore, when~\eqref{eqn:expectation-z} is satisfied, and $\ns>C\dims/\eps^2$
\begin{align}
\probof{Z(\lamb)<\frac{\eps}2}
= &\probof{\expectation{Z(\lamb)}-Z(\lamb)>\expectation{Z(\lamb)}-\frac{\eps}2}\nonumber\\
\le &\probof{\expectation{Z(\lamb)}-Z(\lamb)>\frac{\eps}2}\nonumber\\
\le &\exp\Paren{-\frac{2(\eps/2)^2}{n\cdot (14/\ns)^2}} \nonumber\\
= &\exp\Paren{-\frac{\ns\eps^2}{392}}\nonumber\\
\le & \exp\Paren{-\frac{C\dims}{392}},
\end{align}
proving Theorem~\ref{thm:emp-tv-unif}.

\end{proof}

%

\subsection{Lower bound for $\renprm\ge1$}

In this section we show that empirical estimation of entropy requires at least quadratic (in $\dims)$ samples for $\renprm\ge 1$. This shows that the analysis of empirical estimation is tight in the dimensionality. 

\begin{theorem}
There is a constant $\eps_0$, such that for $\eps<\eps_0$, the empirical estimate of entropy requires $\Omega(\dims^2/\eps)$ samples to estimate von Neumann entropy, and any R\'enyi entropy of order greater than one. 
\label{thm:lb-non-int-large}
\end{theorem}

The proof uses two claims.
The second is on the monotonicity of R\'enyi entropy (See~\cite{BeckS93}).
\begin{lemma}
$\rentprmmst$ is a non-increasing function of $\renprm$.
\label{lem:monotonicity}
\end{lemma}

We show that when the state is maximally mixed, namely each eigenvalue is $\frac1\dims$, the empirical estimator cannot estimate the entropy unless we have enough samples. 
\begin{lemma}
\label{lem:emp-ent}
For any $\renprm\ge1$, there exist constants $c_1$, and $c_2$ such that when the distribution is maximally mixed, and $n<c_1\dims^2/\eps^2$, with probability at least $1-\exp(-c_2\dims^2)$, 
\[
{ S_{\renprm}\Paren{\frac{\lamb}{\ns}}<\log\dims - \eps^2}.
\]
\end{lemma}
\begin{proof}
Suppose $u$ denotes the maximally mixed state. Then, 
\begin{align}
\ent{u}-\empents{\lamb} 
= & \sum_{i=1}^{\dims} \Paren{\frac{\lambi{i}}{\ns}\log\dims -\frac{\lambi{i}}{\ns}\log\frac{\ns}{\lambi i}}\nonumber\\
= & \sum_{i=1}^{\dims} \Paren{\frac{\lambi{i}}{\ns}\log\frac{\ns/{\lambi i}}{\dims}}\nonumber\\
=& d_{KL}\Paren{\frac{\lamb}{\ns}, u}\nonumber\\
\ge& 2\dtv{\frac{\lamb}{\ns}}{u}^2,\label{eqn:pinsker}
\end{align}
where the last step is from Lemma~\ref{lem:distance-bounds}. By Lemma~\ref{lem:monotonicity}, whenever the total variation of the empirical distribution from the uniform distribution is at least $\eps$, all R\'enyi entropies of order at least one at $\Omega(\eps^2)$ away from $\log\dims$. Combining with Theorem~\ref{thm:emp-tv-unif} gives the result. 
\end{proof}

\begin{proof}[Proof of Theorem~\ref{thm:lb-non-int-large}]
The last lemma says that to estimate the entropy to $\pm\eps^2$, the empirical estimate requires $\Omega(\dims^2/\eps^2)$ samples. Substituting $\eps=\eps^2$ gives the result. 
\end{proof}

%% file: converse-small-alpha-empirical.tex
\subsection{Lower bound for $\renprm<1$}

The lower bounds for the empirical algorithm for $\renprm\ge1$ were shown for the uniform distribution. However, for $\renprm<1$, the uniform distribution only provides a quadratic dependence on $\dims$. This is similar to the classical setting, where~\cite{AcharyaOST17} designed another distribution for analyzing the case $\renprm<1$. We will consider their distributions as our eigenvalues (for ease of notations we use the dimension to be $\dims+1$):
\begin{align}
 \eiglbi{1} =1-\frac{\eps}{\dims^{\frac1{\renprm}-1}}, \eiglbi{i} = {\frac{\eps}{\dims^{\frac1\renprm}}}, \text{ for } i=2,\ldots,\dims+1.\label{dist-lower}
\end{align}

{Thus we assume that $\eps < \dims^{1/\renprm - 1}$.}
\cite{AcharyaOST17} consider $\eiglb$ as a distribution for the classical setting and showed that the empirical plug-in estimator requires $\Omega((\dims/\eps)^{1/\renprm})$ samples to output an $\pm \eps$ estimate of $\rentprmmst$. In particular, they show that unless $\Omega(\dims^{1/\renprm}/\eps)$, the plug-in estimator of Renyi entropy (See Section~\ref{sec:eyd}) $H_{\renprm}(\hat {p})$ {is at most} $S_{\renprm}(\eiglb)-\eps^{\renprm}$. We can now invoke the following inequality from~\cite{HardyLP29}, and~\cite[Equation~(1), Chapter 1]{MarshallOA11}. 
\begin{lemma}
If $f$ is a concave function, and $x_1, \ldots,x_m$ majorizes $y_1, \ldots, y_m$, then 
\[
\sum_{i=1}^{m} f(x_i)\le \sum_{i=1}^{m} f(y_i).
\]
\end{lemma}
Since $x^\renprm$ is a concave function, {and $\lamb/n$ majorizes the profile}, this implies that $S_\renprm(\lamb/\ns)$ is at most $H_\renprm(\hat {p})$. This proves a lower bound of $\Omega((\dims/\eps)^{1/\renprm})$ to estimate $\rentprmmst$ to $\pm\eps^{\renprm}$ in the present setting. We now provide an improved lower bound. Suppose $\mst$ is a density matrix with eigenvalues $\eiglb$. 
\begin{theorem}
The EYD algorithm requires $\Omega(\dims^{1+1/\renprm}/\eps^{1/\renprm})$ copies to estimate $\rentprmmst$ to $\pm\eps$. 
\label{thm:lb-non-int-small}
\end{theorem}

We will consider the LIS interpretation of the Young tableaux, generated by a distribution over $[\dims+1]$ that assigns probability $\eiglbi{i}$ to symbol $i$. Recall that the ordering of eigenvalues does not affect the output distribution of Young tableaux. 

Let $\lamb$ be the Young tableaux generated from $\ns$ independent samples. 
{Fix $\beta > 0$}, and consider the following events. 

\noindent {\bf $\cE_1:$} $\lambi1$ is equal to the number of occurrences of symbol 1. 

\noindent {\bf $\cE_2:$} If $M=\ns-\lambi1=\sum_{i=2}^{\dims+1} \lambi{i}$
then
\begin{align}  
\sum_{i=2}^{\dims +1}\absv{\frac{\lambi{i}}{M}-{\frac{1}{\dims}}}>2\beta.
\end{align} 

\noindent {\bf $\cE_3:$} 
 \[\frac{M^\renprm}{\ns^{\renprm}}\dims^{1-\renprm}< \eps^{\renprm}\left(1+\beta^2\frac{\renprm(1-\renprm)}2\right).\]

The following lemma, proved in Appendix~\ref{sec:app-conditioning}, states that these events occur unless $\ns$ is large. 
\begin{lemma} 
{
If $\ns = \Omega(\dims^{1/\renprm-1}/\eps)$ and
$\ns=O(\dims^{1+1/\renprm}/\eps)$, then with probability at least 0.9, $\cE_1, \cE_2, \cE_3$ all occur.}
\label{lem:conditioning}
\end{lemma}

We now prove Theorem~\ref{thm:lb-non-int-small} assuming Lemma~\ref{lem:conditioning}.  
{By the reasoning before Theorem~\ref{thm:lb-non-int-small}, we may assume that $\ns = \Omega(\dims^{1/\renprm}/\eps)$, and thus $\ns$ is in the range assumed by Lemma~\ref{lem:conditioning}.}
We will use the following lemma, which states that if a distribution is far from the uniform distribution, its $\renprm$th moment is far from that of the uniform distribution. 
\begin{lemma}
\label{lem:tv-renyi-small}
Let $p$ be a distribution over $[\dims]$ such that $\sum_{i=1}^{\dims} |p(i)-\frac1\dims|=2\gamma$, then for $\renprm<1$,
\[
\sum_{i=1}^{\dims}p_i^{\renprm}\le\Paren{1-{\renprm}{(1-\renprm)}\cdot\gamma^2}\cdot\dims^{1-\renprm}.
\]
\end{lemma}
\noindent The lemma is proved in Appendix~\ref{sec:app-far-from-uniform}.
We can now consider $\lambi{i}/M$ for $i=2\upto\dims+1$ as a distribution over $\dims$ elements. Applying the lemma, when $\cE_2$ holds, 
\[
\sum_{i=2}^{\dims+1} \Paren{\frac{\lambi{i}}{M}}^{\renprm} < (1 - {\renprm(1-\renprm)}\beta^2)\dims^{1-\renprm}.
\]
When this happens, 
\begin{align}
\sum_{i=1}^{\dims+1} \Paren{\frac{\lambi{i}}{\ns}}^{\renprm} = & 
\Paren{\frac{\lambi{1}}{\ns}}^{\renprm} + \Paren{\frac{M}{\ns}}^{\renprm}\cdot\Paren{\sum_{i=2}^{\dims+1} \Paren{\frac{\lambi{i}}{M}}^{\renprm}}\nonumber\\
< & 1 + \Paren{\frac{M}{\ns}}^{\renprm}\cdot (1 - {\renprm(1-\renprm)}{\beta^2})\dims^{1-\renprm}\nonumber\\
\le & 1 + \eps^{\renprm}\Paren{1+\beta^2\frac{\renprm(1-\renprm)}{2}}(1 - \beta^2{\renprm(1-\renprm)})\label{eqn:alpha-small-step}\\
\le & 1  + \eps^{\renprm}\Paren{1-\beta^2\frac{\renprm(1-\renprm)}{2}}\label{eqn:elb},
\end{align}
where~\eqref{eqn:alpha-small-step} follows from $\cE_3$. 

\noindent We now relate this to $\Map{\renprm}{\eiglb}$. 
\[
\Map{\renprm}{\eiglb} = \Paren{1-\frac{\eps}{\dims^{\frac1{\renprm}-1}}}^{\renprm} + \dims \cdot \Paren{\frac{\eps}{\dims^{\frac1\renprm}}}^{\renprm}.
\]
For any $x<0.5$, $(1-x)^\renprm>1-2\renprm x$, and therefore, for $\eps<0.1$, 
\begin{align}
\Map{\renprm}{\eiglb} \ge 1- 2\renprm\frac{\eps}{\dims^{\frac1{\renprm}-1}} +\eps^\renprm.\label{eqn:lb-mmt-small}
\end{align}

{Suppose $\dims^{1/\renprm-1}>\frac{10}{\beta^2(1-\renprm)}$}. Then 
\begin{align}
\Map{\renprm}{\eiglb} \ge 1+ \eps^{\renprm}\cdot \Paren{1-\beta^2\frac{\renprm(1-\renprm)}{5}}\label{eqn:mlb}
\end{align}
Now simply comparing~\eqref{eqn:mlb}, and~\eqref{eqn:mlb}, and {using the fact that $\log(1+x) = x + \Theta(x^2)$} for $|x|<1/2$, we note that the EYD R\'enyi entropy is a factor $\Omega(\eps^{\renprm})$ away from the true entropy. Therefore,  by Lemma~\ref{lem:conditioning}, unless $\ns=\Omega(\dims^{1+1/\renprm}/\eps)$ we cannot estimate entropy up to $\pm \eps^\renprm$. Substituting $\eps$ with $\eps^{1/\renprm}$ gives the theorem.

%% file: acknowledgements.tex
\section*{Acknowledgements}

The authors would like to thank John Wright for detailed comments on a manuscript of the paper. In particular, he pointed out a mistake in the proof of the copy complexity lower bound for integral~$\renprm$, which has been fixed in this version. They also thank John Wright for sharing their results on von Neumann entropy estimation.  

%% file: proof-lemma-lengthpartition.tex
\section{Proof of Lemma~\ref{lem:lengthpartition}}
\label{app:lemma-partition}

Fix $j \in \{0 \upto \renprm-1\}$. First we prove that  each cycle in $\sigma = \sigma_1 \circ \sigma_2$ contains an element in $\{j+1,\cdots,\alpha \}$.
Let \begin{align*}
S = \{j+1,\cdots,\alpha \}, \qquad F_1 = \{ \alpha+1,\cdots,\alpha+j\}, \qquad  \text{and } F_2 = \{1, \cdots, j\}.
\end{align*}
$F_1$ is fixed under $\sigma_1$, and $F_2$ is fixed under $\sigma_2$. Now consider any cycle in $\sigma$, and pick an element $k$ in the cycle. If $k \in S$, then the claim is true. Otherwise: \\
\underline{Case 1:} $k \in F_2$. \\
Let $n_k$ be the largest integer such that $\sigma_1(k), \sigma_1^2(k), \dots, \sigma_1^{n_k} (k) \in F_2$. (If $\sigma_1(k) \notin F_2$, define $n_k = 0$.) Note that, since $\sigma_1$ performs a cycle on $F_2 \cup S$, there must exist $m$ such that $\sigma_1^m (k) \in S$. Hence, $n_k$ is finite. Then,
\begin{align*}
\sigma^{n_k+1} (k) = (\sigma_1 \circ \sigma_2)^{n_k+1} (k) 
\stackrel{\text{(a)}} = (\sigma_1 \circ \sigma_2)^{n_k} (\sigma_1 (k))  
\stackrel{\text{(b)}} = \sigma_1^{n_k+1} (k) \stackrel{\text{(c)}} \in S,
\end{align*}
where (a) follows from the definition of $n_k$ and the fact that points in $F_2$ are fixed under $\sigma_2$, (b) follows similarly (by induction), and (c) follows from the definition of $\sigma_1$ and $n_k$. \\
\underline{Case 2:} $k \in F_1$. \\
Let $n_k$ be the largest integer such that $\sigma_2(k), \sigma_2^2(k), \dots, \sigma_2^{n_k} (k) \in F_1$. (If $\sigma_2(k) \notin F_1$, define $n_k = 0$.) Note that, since $\sigma_2$ performs a cycle on $F_1 \cup S$, there must exist $m$ such that $\sigma_2^m (k) \in S$. Hence, $n_k$ is finite. Then,
\begin{align*}
\sigma^{n_k+1} (k) = (\sigma_1 \circ \sigma_2)^{n_k+1} (k) 
\stackrel{\text{(a)}} = (\sigma_1 \circ \sigma_2)^{n_k} (\sigma_2 (k))  
\stackrel{\text{(b)}} = \sigma_1 \circ \sigma_2^{n_k+1} (k),
\end{align*}
where (a) follows from the definition of $n_k$ and the fact that points in $F_1$ are fixed under $\sigma_1$, and (b) follows similarly (by induction). Now, by definition of $\sigma_2$ and $n_k$, $\sigma_2^{n_k+1} (k) \in S$. Then, by definition of $\sigma_1$, $ \sigma_1 \circ \sigma_2^{n_k+1} (k) \in F_2 \cup S$. If $ \sigma_1 \circ \sigma_2^{n_k+1} (k) \in S$, then the claim is true. Finally, if $ \sigma_1 \circ \sigma_2^{n_k+1} (k) \in F_2$, this falls back to case 1 which has been resolved.

Since $|\{j+1,\cdots,\alpha \}|=\alpha-j$ it follows that  $\ell(\mu)\le\alpha-j$.

%% file: app-majorization.tex
\section{Proof of Lemma~\ref{lem:majorization}} \label{sec:app-majorization}
\label{app:lemma-majorization}
Let  $\ell= \ell(\mu_1)= \ell(\mu_2)$, $\mu_1= (x_1, \dots, x_{\ell})$, and $\mu_2 = (y_1, \dots, y_{\ell})$. Then
\begin{align*}
\Mapeigv{\mu_1} = \prod_{i=1}^\ell \Mapeigv{x_i} = \prod_{i=1}^\ell  \sum_{j=1}^{\dims} \eig_{j}^{x_i} = \sum_{ j_1, \dots, j_{\ell} \in [\dims]^\ell} \eta_{j_1}^{x_1} \dots \eta_{j_\ell}^{x_\ell}.
\end{align*}
We define an equivalence relation on $[\dims]^{\ell}$ as follows: $(j_1,\dots,j_\ell) \sim (\hat{j}_1,\dots,\hat{j}_{\ell})$ if there exists a permutation $\sigma$ on $[\ell]$ such that $\sigma (j_1,\dots,j_{\ell}) = (\hat{j}_1,\dots,\hat{j}_{\ell})$. We denote by $\mathcal{E}$ the set of equivalence classes created by this relation, and for each $E \in \mathcal{E}$ we pick a representative element and denote it by $(j_1,\dots,j_{\ell})_E$. For each $E$, define $g_E: E \rightarrow [\ell!]$ as 
\begin{align*}
g_E (j_1, \dots, j_{\ell}) = \Big| \left\lbrace \sigma: \sigma \big((j_1,\dots,j_{\ell})_E \big) = (j_1,\dots,j_{\ell}) \right\rbrace \Big|, \quad (j_1,\dots,j_{\ell}) \in E.
\end{align*}
Now note that, for each $E$, $g_E(.)$ is a constant function. Indeed, if $(j_1,\dots,j_{\ell})$ and $(\hat{j}_1,\dots,\hat{j}_{\ell})$ belong to $E$, then there exists $\sigma_1$ such that $\sigma_1(j_1, \dots, j_{\ell}) = (\hat{j}_1,\dots,\hat{j}_{\ell})$. Therefore if $\sigma ((j_1,\dots,j_{\ell})_E) = (j_1,\dots,j_{\ell})$, then $\sigma_1 \circ \sigma ((j_1,\dots,j_{\ell})_E) = (\hat{j}_1,\dots,{j}_{\ell})$. Similarly, if  $\sigma ((j_1,\dots,j_{\ell})_E) = (\hat{j}_1,\dots,\hat{j}_{\ell})$, then $\sigma_1^{(-1)} \circ \sigma ((j_1,\dots,j_{\ell})_E) = (j_1,\dots,j_{\ell})$. So define $g: \mathcal{E} \rightarrow [\ell!]$ as $g(E) = g_E( (j_1,\dots,j_{\ell})_E).$ Now consider
\begin{align*}
\Mapeigv{\mu_1} & = \sum_{ j_1, \dots, j_{\ell} \in [\dims]^\ell} \eta_{j_1}^{x_1} \dots \eta_{j_{\ell}}^{x_{\ell}} \\
& = \sum_{E \in \mathcal{E}} \sum_{ (j_1,\dots,j_{\ell}) \in E} \eta_{j_1}^{x_1} \dots \eta_{j_{\ell}}^{x_{\ell}} \\
& = \sum_{E \in \mathcal{E}} \frac{1}{g(E)} \sum_{ \sigma} \eta_{\sigma(j_{1,E})}^{x_1} \dots \eta_{\sigma(j_{\ell,E})}^{x_{\ell}} \\
& \geq \sum_{E \in \mathcal{E}} \frac{1}{g(E)} \sum_{ \sigma} \eta_{\sigma(j_{1,E})}^{y_1} \dots \eta_{\sigma(j_{\ell,E})}^{y_{\ell}} \\
& = \Mapeigv{\mu_2}, 
\end{align*}
where the inequality follows from Muirhead's theorem~\cite{muirhead1902}\cite[p. 125]{MarshallOA11}.

%% file: app-falling.tex
\newest{\section{Proof of Lemma~\ref{lem:falling}}
\label{app:lemma-falling}
Let $|\mu|=qd+r$, where $0<r<d$, and $q$ are non-negative integers. 
We will show that of all $\mu$ with $|\mu|=qd+r$ and $\ell(\mu)\le\dims$, the tableaux that has $q$ columns with $\dims$ boxes, and one last column with $r$ boxes, minimizes $d^{\overline{\mu}}$. 
Toward this end consider a tableaux that has at least two non-empty columns that have less than $\dims$ boxes in them. Then we can move a box from the last row with length $\mu_1$, and move it to the end of the first column that does not have length equal to $\dims$. This operation moves a box to the left and below, thereby decreasing the value of $c(\square)$ for it.}

\newest{We now assume that the partition $\mu$ has $q$ columns with $\dims$ boxes and one column with $r$ boxes. For this partition $\mu$, by Definition~\ref{def:falling},
\begin{align}
\dims^{\overline{\mu}} = & (d+q)^{\underline{r}} \prod_{j=0}^{q-1}(\dims+j)^{\underline d} \nonumber\\
\ge & (\dims!)^q\cdot (\dims)^{\underline{r}}.\label{eqn:falling}
\end{align}
We will show that for any integer $0\le t\le d$, 
\begin{align}
d^{\underline{t}}\ge\Paren{\frac \dims e}^t.\label{eqn:fall-bound}
\end{align}
Plugging this bound in~\eqref{eqn:falling}, and noting that $d! = d^{\underline{d}}$, we obtain, 
\begin{align}
\dims^{\overline{\mu}} \ge \Paren{\Paren{\frac \dims e}^d}^{q}\cdot \Paren{\frac{\dims}{e}}^r =\Paren{\frac \dims e}^{|\mu|}
\end{align}
We now prove~\eqref{eqn:fall-bound}. Let 
\[
f(t) = \frac{d^{\underline{t}}}{\Paren{\frac \dims e}^t}.
\]
Then for any $t \le d$, 
\[
\frac{f(t+1)}{f(t)} = \frac{(d-t)}{(d/e)},
\]
and this ratio is monotonically decreasing with $t$. Therefore, the smallest value of $f(t)$ occurs at either $t = 0$ or $t = d$. At $t=0$,~\eqref{eqn:fall-bound} is true since both sides are 1, and at $t=d$, we need to show that
\[
d!>\Paren{\frac de}^d,
\] 
which follows from Stirling's approximation.
}

%% file: app-bias-von-neumann.tex
\section{Proof of Lemma~\ref{lem:opt-problem}}
\label{sec:app-bias}

We will first show that at the maxima, there can be at most three distinct values that the $x_i$'s can take, of which at most one is positive. 

\noindent Consider $x_i>0$ and $x_j>0$. Then, by the concavity of logarithm, if we replace both of them by $(x_i+x_j)/2$ the objective value increases. The constraints, on the other hand, are still valid.

\noindent We now consider the negative values. Writing the Lagrangian, 
\begin{align}
\cL\Paren{x_1\upto x_\dims, \gamma_1, \gamma_2} = \sum_{i=1}^\dims \Paren{x_i \log \frac1{x_i^2}} + \gamma_1\Paren{4-\sum_i x_i^2} +\gamma_2\Paren{\sum_i x_i}.
\end{align}
Differentiating with respect to $x_i$, 
\[
\frac{\partial \cL\Paren{x_1\upto x_\dims, \gamma_1, \gamma_2}}{\partial x_i}
= -\log \frac1{x_i^2} - 2 -2\gamma_1 x_i +\gamma_2 = 0,
\]
and therefore, 
\[
2\gamma_1 x_i + \log\frac1{x_i^2} -\gamma_2 -2=0
\]
This function is strictly convex on $(-\infty,0)$, and therefore has at most two roots. 

Therefore, there are at most three distinct values that $x_i$'s can take, and at most one of them is positive. Let $y_1>0>-y_2>-y_3$ be these values, and let $d_1, d_2, d_3$ be the multiplicities of these. 
Therefore, the optimization problem can be written as:
\begin{align}
{\bf P3:}\qquad &  \text{maximize}\ \ d_1 y_1 \log \frac1{y_1^2}-d_2 y_2 \log \frac1{y_2^2}-d_3 y_3 \log \frac1{y_3^2} \nonumber\\
&  \text{subject to}\ 
d_1 y_1 -d_2y_2-d_3y_3 =0, \ d_1 y_1^2+d_2 y_2^2+d_3 y_3^2\le 4, \text{ and }\ d_1+d_2+d_3 \le \dims.\nonumber
\end{align}
Substituting $ d_1 y_1 =d_2y_2+d_3y_3$ the objective becomes
\begin{align}
\Paren{d_2y_2+d_3y_3}\log \frac{1}{y_1^2}-d_2 y_2 \log \frac1{y_2^2}-d_3 y_3 \log \frac1{y_3^2} 
= d_2 y_2 \log \frac{y_2^2}{y_1^2}+d_3 y_3 \log \frac{y_3^2}{y_1^2}.\nonumber
\end{align}
Since $d_1y_1\ge d_2y_2$, we have $y_2/y_1 \le d_1/d_2$, and 
\begin{align}
d_2 y_2 \log \frac{y_2^2}{y_1^2} 
\le 
2d_2 y_2 \log\frac{d_1}{d_2} = 2 \sqrt {d_1}\cdot\Paren{\sqrt{\frac{d_2}{d_1}}\log\frac{d_1}{d_2}}\cdot  \Paren{\sqrt {d_2}y_2}.\nonumber
\end{align}
Since $d_2y_2^2\le4$,  $\sqrt {d_2}y_2<2$. Moreover, for any $z>0$, 
\[
z\log \frac1{z^2} = 2z\log\frac1z \le \frac 2e.
\]
This shows that 
\[
d_2 y_2 \log \frac{y_2^2}{y_1^2} \le \frac 8e \sqrt {d_1} \le \frac 8e\sqrt \dims.
\]
By a similar argument, 
\[
d_3 y_3 \log \frac{y_3^2}{y_1^2}  \le \frac 8e\sqrt \dims.
\]
Summing up the two terms bounds the objective of {\bf P3}, and plugging in~\eqref{eqn:opt-problem}, we obtain 
\[
\absv{\sum_{i=1}^{\dims} (\eigvi{i} - \expectation{\eighati{i}})\log\frac{1}{\eigvi{i}} } \le \frac {16}e\sqrt \dims\cdot\sqrt{\frac\dims\ns} = \frac {16}e \frac{\dims}{\sqrt\ns}.
\]

%% file: app-lipschitz-entropy.tex
\section{Proof of Lemma~\ref{lem:lipschitz}}
\label{sec:app-lipschitz-entropy}
\begin{proof}
In the classical setting, changing one element can change at most two probabilities of the empirical distribution, using which one can bound the variance of the empirical entropy estimator~\cite{paninski2003}. However, in our case, changing one symbol can change the length of more than one of the rows of the Young tableaux. \cite[Prop.~2.2]{OW17} showed that the cumulative row sums are bounded (see also \cite{BhatnagarL12}). In particular, \textit{for any $j=1, \ldots, d$,}
\[
\absv{\sum_{i=1}^{j} \lambi{i} - \sum_{i=1}^{j} \lamb'_i} \le1.
\]
Suppose $\Delta_i \ed\lambi{i}'-\lamb_i$, then for all $j=1,\ldots,\dims$,
\begin{align}
-1\le \sum_{i=1}^{j} \Delta_i \le1.\label{eqn:cum-sum}
\end{align}
This also implies that the for each $i$, $-2\le \Delta_i\le 2$. This proves a bounded difference condition on $\lambi{i}$, which can be used to prove its concentration using McDiarmid's inequality (Lemma~\ref{lem:mcdiarmid}). 

Note that $\lambi{i}$ changes by at most two when one of the inputs changes, and hence $c=2$. This gives
\begin{align}
\probof{\absv{\lambi{i}-\expectation{\lambi{i}}}>t}\le 2\cdot e^{-\frac{t^2}{2n}}.
	\label{eqn:conc-lambi}
\end{align}

Without loss of generality assume that $\ell(\lamb)\ge\ell(\lamb')$, i.e., the number of rows in $\lamb$ is at least the number of rows in $\lamb'$. 

By the Taylor series, for any $-x\le \delta\le x$,
\begin{align}
(x+\delta)\log(x+\delta) = x \log x+ \delta(1+\log x) + \sum_{j=2}^{\infty} \frac{\delta^j}{(j-1)j}\frac{(-1)^j}{x^{j-1}}.
\end{align}

Let $f(x) = x\log x$. Then
\begin{align}
&\empents{\lamb'}-\empents{\lamb}\nonumber\\
=&  \sum_{i=1}^{\ell(\lamb)}\frac{\lamb'_i}{\ns}\log\frac{\ns}{\lamb'_i}- \frac{\lambi{i}}{\ns}\log\frac{\ns}{\lambi{i}}  \nonumber\\
=& \sum_{i=1}^{\ell(\lamb)}   - f\Paren{\frac{\lambi{i}+\Delta_i}{\ns}} + f\Paren{\frac{\lambi{i}}{\ns}}\nonumber\\
=& \sum_{\lambi{i}>1} - \Brack{\frac{\Delta_{i}}{\ns}\left(1+\log\frac{\lambi{i}}{\ns}\right)+\frac1{\ns}\sum_{j=2}^{\infty} \Paren{\frac{\Delta_i}{\lambi{i}}}^{j-1}\frac{\Delta_i (-1)^j}{(j-1)j}} 
+ \sum_{\lambi{i}=1} \Brack{\frac{1+\Delta_i}{\ns}\log\frac{\ns}{1+\Delta_i} - \frac1{\ns}\log\ns}. \nonumber
\end{align}

We now consider the terms separately, and prove the following series of (simple) claims. 
\begin{enumerate}
\item If~\eqref{eqn:cum-sum} holds, and $\lambda_i$'s are non-increasing, then 
\[
\absv{\sum_{\lambi{i}>1} \Brack{\frac{\Delta_{i}}{\ns}\Paren{1+\log\frac{\lambi{i}}{\ns}}}}
\le \absv{\sum_{\lambi{i}>1} {\frac{\Delta_{i}}{\ns}}}+\absv{\sum_{\lambi{i}>1} {\frac{\Delta_{i}}{\ns}\log\frac{\lambi{i}}{\ns}}}  \le \frac1\ns + \frac2{\ns}\log \frac{n}{2}.
\]

\noindent\textit{Proof.} The first term is a direct consequence of~\eqref{eqn:cum-sum}. The second term follows from the following lemma.
\begin{lemma}
Let $x_1 \ge x_2 \cdots \ge x_m\ge 0$ be real numbers. Further, let $\Delta_1, \ldots, \Delta_m$ satisfy $-1\le\sum_{i=1}^j \Delta_i \le 1$. Then $\sum_{i=1}^m \Delta_i x_i \le x_1$. 
\end{lemma}

\item 
For $\lambda\in\NN$, let $I_{\lambda}=\{i:\lambi{i} = \lambda\}$ be the set of rows with length $\lambda$. Then, \[
\absv{\{i\in I_\lambda: \Delta_i\ne0\}}\le 4,
\]
\textit{i.e.}, there are at most four non-zero $\Delta_i$'s for each distinct value of $\lambi i$. 

\noindent\textit{Proof.} Let $\lambi{i}= \lambda$ for $i \in\{h_1, \ldots, h_2\}$.
However, since the $\lamb'_{i}$'s are non-increasing, $\Delta_i$'s are 
non-increasing
for all $i\in\{h_1,\ldots, h_2\}$. If more than four of these are non-zero, then there are at least three consecutive positive, or three consecutive negative $\Delta_i$'s. However, this would violate~\eqref{eqn:cum-sum}. 
\item
There are at most $\sqrt{2\ns}$ non-empty $I_{\lambda}$'s.

\noindent\textit{Proof.} The sum of row lengths equals $\ns$. Therefore, the maximum number of distinct row lengths is the largest value of $j$ such that $\sum_{i=1}^{j} i = j(j+1)/2$ is at most $\ns$. This proves that $j$ is at most $\sqrt {2\ns}$. 
\item
For $\lambi i\ge2$, 
\[
\absv{\sum_{j=2}^{\infty} \Paren{\frac{\Delta_i}{\lambi{i}}}^{j-1}\frac{\Delta_i(-1)^j}{(j-1)j}} \le 2.
\]

\noindent\textit{Proof.} Using $\absv{\Delta_i}\le 2$, and $\lambi i\ge2$, we have $\absv{\Delta_i}/\lambi i\le 1$. This implies that for any $j\ge1$, $(\absv{\Delta_i}/\lambi i)^j\le \absv{\Delta_i}/\lambi i$. This gives
\begin{align}
\absv{\sum_{j=2}^{\infty} \Paren{\frac{\Delta_i}{\lambi{i}}}^{j-1}\frac{\Delta_i(-1)^j}{(j-1)j}}\le& \sum_{j=2}^{\infty} \Paren{\frac{\absv{\Delta_i}}{\lambi{i}}}^{j-1}\frac{\absv{\Delta_i}}{(j-1)j}\nonumber \\
\le & \Paren{\frac{\absv{\Delta_i}}{\lambi{i}}}\cdot\sum_{j=2}^{\infty} \frac{\absv{\Delta_i}}{(j-1)j}\nonumber\\
= & \frac{\Delta_i^2}{\lambi i} \le 2.\nonumber
\end{align}
\item The second summation satisfies
\[
\absv{\sum_{\lambi{i}=1} \Brack{\frac{1+\Delta_i}{\ns}\log\frac{\ns}{1+\Delta_i} - \frac1{\ns}\log\ns}} \le \frac{8\log\ns}{n}
\]
whenever $\ns \ge 27$.

\noindent\textit{Proof.}
\begin{align}
\absv{\sum_{\lambi{i}=1} \Brack{\frac{1+\Delta_i}{\ns}\log\frac{\ns}{1+\Delta_i} - \frac1{\ns}\log\ns}} =& 
\absv{\sum_{\lambi{i}=1,\Delta_i \ne 0} \Brack{\frac{1+\Delta_i}{\ns}\log\frac{\ns}{1+\Delta_i} - \frac1{\ns}\log\ns}} \nonumber\\
\le & \sum_{\lambi{i}=1,\Delta_i \ne 0} \absv{\Brack{\frac{1+\Delta_i}{\ns}\log\frac{\ns}{1+\Delta_i} - \frac1{\ns}\log\ns}} \nonumber\\
\le & 4\absv{\Brack{\frac{3}{\ns}\log\frac{\ns}{3} - \frac1{\ns}\log\ns}} \nonumber\\
\le &\frac{8\log\ns}{n},
\end{align}
where the middle inequality holds whenever $\ns \ge 27$.
\end{enumerate}
Using these five simple claims, we can bound the difference between $\empents{\lamb}$ and $\empents{\lamb'}$.
\end{proof}

%% file: app-lipschitz-tv.tex
\section{Proofs of Lemmas~\ref{lem:dtv-lipschitz} through~\ref{lem:emp-quant}}
\label{app:lipschitz}

\begin{proof}[Proof of Lemma~\ref{lem:dtv-lipschitz}]
The proof follows the same tools as that of proving the Lipschitzness of empirical entropy. 
By the relation between total variation and $\ell_1$ distance, we know that 
\[
Z(\lamb) = 2\cdot \sum_{i:\frac{\lambi{i}}{\ns}>\frac1\dims}\frac{\lambi{i}}{\ns}-\frac1\dims.
\]
Let $j$, and $j'$ be the largest indices such that $\lambi{j}/n>1/\dims$, and $\lamb{'}_{j'}/n>1/\dims$. From the fact that for any $\ell$, $-1\le \sum_{i=1}^{\ell} (\lambi{i}-\lamb{'}_i)\le1$, we conclude that $|j'-j|\le 3$. Moreover, at each of these three locations the value of $\lamb{'}_i$ and $\lambi{i}$ differ by at most two. Combining these results proves the bound. 
\end{proof}
\begin{proof}[Proof of Lemma~\ref{lem:l-one-error}]
Define
\[
Z(\hat p) \ed \sum_{i=1}^{\dims} \absv{\hat{p}(i)-\frac1\dims}.
\]
We will use following bound on the expected value of $Z(\hat p)$ for the uniform distribution. 
\begin{lemma} 
\label{lem:expected-lone}
\[
\expectation{Z(\hat p)}\ge \min\left\{\frac{1}{2}, \frac{8}{81}\sqrt{\frac{\dims}{\ns}}\right\}.
\]
\end{lemma}
Assuming this result, we can prove Lemma~\ref{lem:l-one-error} in the following two cases. 
\paragraph{Case 1: $\ns < \dims/4$.} In this case, at least $3\dims/4$ $\hat p(i)$'s are zero. Therefore, with probability one, 
\[
Z(\hat p)\ge \frac{3\dims}4\cdot \frac1{\dims} = \frac 34.
\]
\paragraph{Case 2: $\ns \ge \dims/4$.} In this case, the second term in Lemma~\ref{lem:expected-lone} controls the minimum. Suppose $\eps<1/(2c_1)$. We can find a constant $c_3$ such that when $\ns<c_3\dims/\eps^2$, $\expectation{Z(\hat p)}>2\eps$, which implies $\expectation{Z(\hat p)}-\eps>\expectation{Z(\hat p)}/2$. Applying McDiarmid's inequality on $\hat p$, noting that changing one symbol changes $Z(\hat p)$ by at most $2/\ns$, 
\begin{align}
\probof{Z(\hat p)<{\eps}}
= &\probof{\expectation{Z(\hat p)}-Z(\hat p)>\expectation{Z(\hat p)}-{\eps}}\nonumber\\
\le &\probof{\expectation{Z(\hat p)}-Z(\hat p)>\frac{\expectation{Z(\hat p)}}{2}}\nonumber\\
\le &\exp\Paren{-\frac{2\Paren{\frac{\expectation{Z(\hat p)}}{2}}^2}{n\cdot (2/\ns)^2}} \nonumber\\
= & \exp\Paren{-\frac{n\cdot{{\expectation{Z(\hat p)}}}^2}{8}} \nonumber\\
\le & \exp\Paren{-\frac{\dims}{8c_2^2}}, \nonumber
\end{align}
where the last step used that $\expectation{Z(\hat p)}\ge\frac1{c_2}\sqrt{\frac{\dims}{\ns}}$.
\end{proof}

\begin{proof}[Proof of Lemma~\ref{lem:expected-lone}]
$\ns \hat p(i)$ is distributed ${\rm Bin}(\ns, \frac1\dims)$. Suppose $X\sim  {\rm Bin}(\ns, \frac1\dims)$, then by the linearity of expectations,
\begin{align}
\expectation{Z(\hat p)} = {\dims}\cdot\expectation{\absv{\frac X{\ns}-\frac{1}{\dims}}}
= \frac{\dims}{\ns}\cdot\expectation{\absv{X-\frac{\ns}{\dims}}}.\label{eq:l-one}
\end{align}
There is clean expression for the expected absolute deviations of Binomial random variables~\cite[Lemma 1.4]{mattner2003mean}:
\begin{lemma}[De Moivre’s mean absolute deviation identity]
\label{lem:demoivre}
Let $b(n, p; k) = {\ns \choose k}p^{k}(1-p)^{\ns-k}$ be the Binomial probability, and let $Y\sim{\rm Bin}(\ns, p)$. Then
\[
\expectation{\absv{X-{\ns p}}} = 2\ns p(1-p)\max_k b(\ns-1, p;k).
\] 
\end{lemma}
We will be interested in $p = \frac1\dims$, and would like to bound $\max_k b(\ns-1, p;k)$. It is possible to use Stirling's approximation for specific values of $k$, but we find it easier to simply apply Chebychev's inequality. By Chebychev's inequality, for any $a>0$, 
\[
\probof{\absv{X-\frac{\ns}{\dims}}>a\cdot \sqrt{\ns\frac1\dims\Paren{1-\frac1{\dims}}}}<\frac1{a^2}.
\]
Suppose $a=3$. Then
\[
\probof{\absv{X-\frac{\ns}{\dims}}> 3 \cdot \sqrt{\frac\ns\dims}}<\frac1{9}.
\]
When $\ns \le \dims/2$, then with probability one, $Z(\hat p)\ge\frac12$.
When $\ns>\dims/2$, then at least one integer lies in the interval $\left[\frac{\ns}{\dims}- 3 \cdot \sqrt{\frac\ns\dims}, \frac\ns\dims+ 3 \cdot \sqrt{\frac\ns\dims}\right]$. The number of integers in this interval is at most $6\sqrt{\frac\ns\dims}+1<9\sqrt{\frac\ns\dims}$. Therefore, there is some $\ell$ such that 
\[
\probof{X=\ell}> \frac89\cdot \frac1{9}\sqrt{\frac\dims\ns}.
\]
We can plug in this expression into Lemma~\ref{lem:demoivre} to obtain
\[
\expectation{\absv{X-{\frac\ns \dims}}} \ge 2\ns\frac1\dims\left(1-\frac1\dims\right)\cdot \frac8{81}{\sqrt{\frac\dims \ns}}\ge \frac8{81}\cdot \frac{\ns}{\dims}\cdot \sqrt{\frac{\dims}{\ns}}.
\]
Plugging this in~\eqref{eq:l-one} proves the lemma.
\end{proof}

%

\begin{proof}[Proof of Lemma~\ref{lem:emp-quant}]
Recall the LIS interpretation of WSS. Suppose the underlying distribution is uniform, namely the state is maximally mixed. Let $\hat{p}$, be the \emph{sorted} plug-in distribution, and $\lamb/n$ is the EYD distribution. Then, Lemma~\ref{lem:emp-eyd} states that $\lamb/\ns$ majorizes $\hat{p}$. Let $j$ be largest index such that $\hat{p}(i)>1/\dims$. Then, 
\begin{align}
Z(\hat p)  = &2\cdot \left(\sum_{i=1}^{j} {\hat{p}(i)-\frac1\dims}\right)\nonumber\\
\le &2 \cdot \Paren{\sum_{i=1}^{j} \absv{\frac{\lambi{i}}{\ns}-\frac1\dims}}\label{eq:maj}\\
\le &Z(\lamb),
\end{align}
where~\eqref{eq:maj} follows from Lemma~\ref{lem:emp-eyd}, and the last step uses Definition~\ref{def:distances}. Invoking Lemma~\ref{lem:l-one-error} proves the claim. 
\end{proof}

%% file: app-large-alpha.tex
\section{Proof of Lemma~\ref{lem:diff-moments}}
\label{sec:large-alpha}

We first show two results that will be used to prove Lemma~\ref{lem:diff-moments}. The first is a simple application of the mean value theorem.
\begin{lemma}
Let $\beta>1$, and $x> y>0$. Then
\[
x^{\beta}-y^{\beta}\le \Paren{x-y}\beta x^{\beta-1}.
\]
\label{lem:diff-moment}
\end{lemma}
\noindent\textit{Proof.} 
By the mean value theorem, 
\[
\frac{x^{\beta}-y^{\beta}}{x-y}\le \max_{z\in[y,x]} \frac{d z^{\beta}}{dz}
\le \beta x^{\beta-1}. \eqed
\]

The next result bounds the moments of a random variable that has exponential tail decay probability. 
\begin{lemma}
Suppose $Z$ is a random variable such that 
\begin{align}
\probof{\absv{Z}>t}<2e^{-t^2/2n},\nonumber
\end{align}
then for any $\beta>0$, there is a constant $C_\beta$ such that
\[
\expectation{|Z|^\beta} < C_\beta n^{\beta/2}.
\]
\label{lem:bounding-moments}
\end{lemma}
\begin{proof}
For a non-negative random variable $X$,
\[
\expectation{X} = \int_0^{\infty} \probof{X>t}dt.
\]
Using this for $|Z|^\beta$, 
\begin{align}
\expectation{|Z|^{\beta}} &= \int_{0}^{\infty} \probof{\absv{Z}^{\beta}>t^{\beta}}dt^{\beta}\nonumber\\
& = \int_{0}^{\infty}\beta t^{\beta-1} \probof{\absv{Z}>t}dt\nonumber\\
& \le \int_{0}^{\infty}2\beta t^{\beta-1} e^{-t^2/2n}dt.\nonumber
\end{align}
Using the transformation $u = t^2/2n$, we have $dt = \sqrt{2n/u}du/2$. This implies 
\begin{align}
\expectation{|Z|^{\beta}}
& \le \int_{0}^{\infty}\beta \Paren{\sqrt{2nu}}^{\beta-1} e^{-u}\sqrt{\frac{2n}u}du.\nonumber\\
& = \beta\Paren{{2n}}^{\beta/2}\int_{0}^{\infty}u^{\beta/2-1}  e^{-u}du\nonumber\\
& = \Paren{2^{\beta/2}\beta\Gamma\Paren{\frac\beta2}}{n}^{\beta/2},\nonumber
\end{align}
where $\Gamma(x)$ is the Gamma function. Choosing $C_\beta = 2^{\beta/2}\beta\Gamma(\frac\beta2)$ proves the lemma.
\end{proof}

\begin{proof}[Proof of Lemma~\ref{lem:diff-moments}]
Theorem 1.4 of~\cite{OW17} implies that
\[
\eigvi{i}\ns -2\sqrt{\ns}\le\expectation{\lambi{i}} \le  \eigvi{i}\ns +2\sqrt{\ns}.
\]
Let $\expectation{\lambi{i}} = \eigvi{i}\ns +\diffmean$, then $\absv{\diffmean}\le 2\sqrt\ns$. 
By the triangle inequality, 
\begin{align}
\expectation{\absv{\lambi{i}^{\beta}-(\eigvi{i}\ns)^\beta}}
\le\ & \expectation{\absv{\lambi{i}^{\beta}-\Paren{\expectation{\lambi{i}}}^{\beta}}+\absv{\Paren{\expectation{\lambi{i}}}^{\beta}-(\eigvi{i}\ns)^\beta}}\nonumber\\
=\ & \expectation{\absv{\lambi{i}^{\beta}-\Paren{\eigvi{i}\ns +\diffmean}^{\beta}}}+\absv{\Paren{\eigvi{i}\ns +\diffmean}^{\beta}-(\eigvi{i}\ns)^\beta}.\nonumber
\end{align}
Consider the second term. By Lemma~\ref{lem:diff-moment}, and $|B|<2\sqrt\ns$, 
\begin{align}
\absv{\Paren{\eigvi{i}\ns +\diffmean}^{\beta}-(\eigvi{i}\ns)^\beta} 
\le &\ \beta \cdot 2\sqrt\ns \Paren{\eigvi{i}\ns +2\sqrt\ns}^{\beta-1}\nonumber\\
< &\ \beta2^{\beta}\cdot  \Paren{\sqrt\ns\cdot (\eigvi{i}\ns)^{\beta-1}+\Paren{2\sqrt\ns}^{\beta}}.\nonumber
\end{align}
 Let $Z_i = \lambi{i}-\expectation{\lambi{i}} = \lambi{i}-\eigvi{i}\ns -\diffmean$. Then $Z_i$ is a zero mean random variable. Recall from~\eqref{eqn:conc-lambi},
$\probof{\absv{Z_i}>t}\le e^{-t^2/2\ns}.$
Again applying Lemma~\ref{lem:diff-moment} to the first term, 
\begin{align}
\expectation{\absv{\lambi{i}^{\beta}-\Paren{\eigvi{i}\ns +\diffmean}^{\beta}}} 
= &\ \expectation{\absv{\Paren{Z_i+\eigvi{i}\ns +\diffmean}^{\beta}-\Paren{\eigvi{i}\ns +\diffmean}^{\beta}}} \nonumber\\
\le  & \ \beta\expectation{\absv{Z_i}\Paren{\absv{Z_i}+\eigvi{i}\ns +\absv{B}}^{\beta-1}}\nonumber\\
\le & \ \beta 3^{\beta-1} \expectation{\absv{Z_i}\Paren{\absv{Z_i}^{\beta-1}+(\eigvi{i}\ns)^{\beta-1} +\absv{B}^{\beta-1}}}\nonumber\\
= & \ \beta 3^{\beta-1} \left(\expectation{\absv{Z_i}^{\beta}}+(\eigvi{i}\ns)^{\beta-1}\cdot\expectation{\absv{Z_i}} +\absv{B}^{\beta-1} \cdot\expectation{\absv{Z_i}} \right)\nonumber
\end{align}
Applying Lemma~\ref{lem:bounding-moments} to $Z_i$ proves the lemma. 
\end{proof}

%% file: app-small-alpha.tex
\section{Proof of Lemma~\ref{lem:diff-moment-small-alpha}}
\label{sec:app-small-beta}
For $x_1, x_2\ldots x_n\in\RR$ such that $\sum_ix_i\ge0$ and $\beta<1$, 
\begin{align}
(x_1+\ldots +x_n)^\beta \le \absv{x_1}^\beta+\ldots+\absv{x_n}^\beta.\label{eqn:smallbeta}
\end{align}
Using this and the notation from the proof of Lemma~\ref{lem:diff-moments}, we first show that 
\[
\absv{\lambi{i}^{\beta}-(\eigvi{i}\ns)^\beta} = \absv{{\Paren{Z_i+ \eigvi{i}\ns +B}^{\beta}-(\eigvi{i}\ns)^\beta}} \le |Z_i|^{\beta}+\absv{B}^{\beta}.
\]
If $Z_i+ \eigvi{i}\ns +B \ge \eigvi{i}\ns$, then this follows
by \eqref{eqn:smallbeta} because
\[
\Paren{Z_i+ \eigvi{i}\ns +B}^{\beta}\le (\eigvi{i}\ns)^\beta +|Z_i|^{\beta}+\absv{B}^{\beta}.
\]
If $Z_i+ \eigvi{i}\ns +B \le \eigvi{i}\ns$, then this follows
by \eqref{eqn:smallbeta} because
\[
(\eigvi{i}\ns)^\beta= \Paren{Z_i+ \eigvi{i}\ns +B -Z_i - B}^{\beta} \le \Paren{Z_i+ \eigvi{i}\ns +B}^{\beta} +|Z_i|^{\beta}+\absv{B}^{\beta}.
\]
Therefore, 
\begin{align}
\expectation{\absv{\lambi{i}^{\beta}-(\eigvi{i}\ns)^\beta}}
\le \expectation{{|Z_i|^{\beta}+\absv{B}^{\beta}}} 
\le C_{\beta}\cdot \ns^{\beta/2},
\end{align}
where we used by Lemma~\ref{lem:bounding-moments} that $\expectation{|Z_i|^{\beta}}<C\ns^{\beta/2}$, and that $\absv{B}\le 2\sqrt\ns$.

%% file: app-lemma-conditioning.tex
\section{Proof of Lemma~\ref{lem:conditioning}}
\label{sec:app-conditioning}

We first show that for $\dims$ large enough, with probability at least 0.98, $\lambi{1}$ is equal to the number of 1's. In other words, the longest non-decreasing subsequence simply corresponds to all the 1's in the sequence. 

The proof uses the following result on the probability that a biased random walk never returns to the origin. 
\begin{lemma}
Let $S_1,\ldots,$ be a biased random walk of length $\ns$ starting at the origin. Let $p>0.99$ be the probability of taking a step to the right. Then with probability at least 0.98 $S_i>0$ for all $i>0$. 
\end{lemma}
\begin{proof}
This is ``the drunkard on the cliff'' problem \cite[Section 2.12]{Jacobs2012Discrete}. The probability is  $1-0.01/0.99>0.98$.
\end{proof}

{Let $\cE_1'$ denote the event that there are more 1's than all other elements combined in $X_i, \ldots, X_\ns$, for all $i\ge1$. By the lemma, it follows
that when $1-\frac{\eps}{\dims^{\frac1{\renprm}-1}}>0.99$,
 $\cE_1'$ has probability at least 0.98. The event $\cE_1'$ in turn
implies that $\lamb_1$ is equal to the number of 1's, as desired.
It follows that on $\cE_1'$ all the rows except the first one are determined by the appearances of the remaining symbols, which are from a uniform distribution.}

We now bound the probability of $\cE_2$. 
We first note that $\lambi{1}$ is always at least the number of occurrences of 1. Therefore, $M$ is at most a Binomial $Bin(\ns, \eps/\dims^{1/\renprm-1})$. Now, if $\ns$ is at most $c\dims^{1+1/\renprm}/\eps$, then $\expectation{M} \le c\cdot\dims^2$. Therefore, by the Chebychev's inequality we can assume that $M<2c\dims^2$ with probability at least 0.98. When we condition on $\cE_1'$, using the fact that the $\eiglbi i$'s are all the same for $i>1$, the distribution of $\lambi{2},\ldots, \lambi{\dims+1}$ is the same as obtained by $M$ independent draws from a uniform distribution over $\dims$ symbols. Therefore, we can invoke Theorem~\ref{thm:emp-tv-unif} to prove that $\cE_2$ happens. 

Finally, $\cE_3$ holds when

\[{M}<\frac{\eps\ns}{\dims^{1/\renprm-1}} (1+\beta^2\renprm(1-\renprm)/2)^{1/\renprm}.\]

Note that $M$ has expected value of at most $\ns\eps/\dims^{1/\renprm-1}$. Therefore, by Binomial concentration bounds, there is a constant $C$ such that when $\ns\eps/\dims^{1/\renprm-1}>C$, the equation above holds with probability at least 0.98. {This only requires that $n>C\dims^{1/\renprm-1}/\eps$, which is guaranteed by the hypotheses of the lemma.}

%% file: app-small-alpha-converse.tex
\section{Proof of Lemma~\ref{lem:tv-renyi-small}}
\label{sec:app-far-from-uniform}
 Let $S=\{i: p_i>\frac1\dims\}$, and $|S|=j$. Then, 
\[
\dtv{p}{u} = \sum_{i\in S}\Paren{p_i -\frac1\dims} = p_S -\frac j\dims.
\]
Let $\bar p$ be a distribution over $[\dims]$ defined as follows. For $i\in S$, let 
$\bar p_i = \frac{p_S}j$, and for $i\in [\dims]\setminus S$, $\bar p_i = \frac{(1-p_S)}{\dims-j}$. Since $\dtv{p}{u} = \gamma$, we have $p_S = \frac j\dims+\gamma$, and for $i\in S$, 
$\bar p_i = \frac1{\dims}+\frac{\gamma}j$, and for $i\in [\dims]\setminus S$, $\bar p_i = \frac1{\dims} -\frac{\gamma}{\dims-j}$. Since probabilities are non-negative, $\dims - j\ge\gamma\dims$, implying that 
\begin{align}
1-\gamma\ge\frac j\dims.\label{eqn:small-renprm-req}
\end{align}

$\bar p$ is a two step distribution. Moreover, all elements with probability larger than $\frac1\dims$ in $p$ have a probability larger than $\frac1\dims$ in $\bar p$, and all elements with probability at most $\frac1\dims$ have probability at most $\frac1\dims$ in $\bar p$. Therefore,  
\[
\dtv{\bar p}{u} = \dtv{p}{u}=\gamma>\eps. 
\]
Note that $x^{\renprm}$ is a concave function in $x$, for $\renprm<1$. Therefore, 
\[
\sum_{i=1}^{\dims}\bar p_i^{\renprm}\ge\sum_{i=1}^{\dims} p_i^{\renprm}.
\]
Now, we bound $\sum_{i=1}^{\dims}\bar p_i^{\renprm}$. 
\begin{align}
\sum_{i=1}^{\dims}\bar p_i^{\renprm} =& \  j\cdot\Paren{\frac1{\dims}+\frac{\gamma}j}^{\renprm} + \Paren{\dims-j}\cdot \Paren{\frac1{\dims} -\frac{\gamma}{\dims-j}}^{\renprm}\nonumber\\
= & \ j^{1-\renprm}\cdot\Paren{\frac j{\dims}+{\gamma}}^{\renprm} + \Paren{\dims-j}^{1-\renprm}\cdot \Paren{\frac{\dims-j}{\dims} -{\gamma}}^{\renprm}\nonumber\\
= &\ \dims^{1-\renprm}\cdot\Paren{\Paren{\frac{j}{\dims}}^{1-\renprm}\cdot\Paren{\frac j{\dims}+{\gamma}}^{\renprm} + \Paren{1-\frac{j}{\dims}}^{1-\renprm}\cdot \Paren{1-\frac{j}{\dims} -{\gamma}}^{\renprm}}.\nonumber
\end{align}
Let $x = \frac j\dims$, then from~\eqref{eqn:small-renprm-req}, $0<x\le 1-\gamma$. Then, 
\begin{align}
\sum_{i=1}^{\dims}\bar p_i^{\renprm} =\dims^{1-\renprm}\cdot\Paren{x^{1-\renprm}\cdot\Paren{x+{\gamma}}^{\renprm} + \Paren{1-x}^{1-\renprm}\cdot \Paren{1-x -{\gamma}}^{\renprm}}.
\end{align}
To complete the proof, we show that the term in the parentheses is at most $1-\renprm(1-\renprm) \gamma^2$.
 Suppose $0<\renprm<1$, and $0<\eps<1$. Then, we show that for $x\in[0,1-\eps]$
\begin{align}
x^{1-\renprm}(x+\eps)^\renprm+(1-x)^{1-\renprm}(1-x-\eps)^\renprm <1-\renprm\Paren{1-\renprm} \eps^2.\label{eqn:calc}
\end{align}
\begin{proof}[Proof of~\eqref{eqn:calc}]
We use the generalized Binomial theorem (Binomial series). For $0<y\le1$, and $\renprm>0$, 
\[
(1-y)^\renprm = \sum_{\ell=0}^{\infty} \frac{\flnpwr{\renprm}{\ell}}{\ell!}\cdot (-y)^{\ell}, 
\]
where $\flnpwr{\renprm}{\ell}= \renprm(\renprm-1)\ldots (\renprm-\ell+1)$ as before denotes the falling powers. 
When $\renprm\in(0,1)$ and $y>0$, note that for all $\ell>1$, the signs of $\flnpwr{\renprm}{\ell}$ and $(-1)^\ell$ are different, implying that all the terms beyond the first $(\ell = 0)$ are at most zero. For $\renprm\in(0,1)$, and $y>0$, truncating beyond two terms above,
\begin{align}
(1-y)^\renprm < 1 -\renprm y - \frac{\renprm(1-\renprm)}2 y^2.\label{eqn:gen-bin-approx}
\end{align}
We now proceed to bound~\eqref{eqn:calc}. 
\begin{align}
 &\ x^{1-\renprm}(x+\eps)^\renprm+(1-x)^{1-\renprm}(1-x-\eps)^\renprm \nonumber\\
 = & \ (x+\eps)\cdot \Paren{\frac{x}{x+\eps}}^{1-\renprm} + (1-x)\cdot\Paren{\frac{1-x-\eps}{1-x}}^{\renprm}\nonumber\\
 = & \ (x+\eps)\cdot \Paren{1-\frac{\eps}{x+\eps}}^{1-\renprm} + (1-x)\cdot\Paren{1-\frac{\eps}{1-x}}^{\renprm}\nonumber\\
 < & \ (x+\eps)\Paren{1 - \frac{(1-\renprm)\eps}{x+\eps}  - \frac{\renprm(1-\renprm)}2 {\frac{\eps^2}{(x+\eps)^2}}} + 
 (1-x)\Paren{1 - \frac{\renprm\eps}{1-x}  - \frac{\renprm(1-\renprm)}2 {\frac{\eps^2}{(1-x)^2}}}\label{eqn:main-step-calc}\\
 = & \ 1 - \eps^2\cdot \frac{\renprm(1-\renprm)}2\Paren{\frac1{x+\eps}+\frac1{1-x}}\nonumber\\
 <& \ 1-{\renprm(1-\renprm)}\cdot\eps^2\label{eqn:final-calc}
\end{align}
where~\eqref{eqn:main-step-calc} uses~\eqref{eqn:gen-bin-approx}, and~\eqref{eqn:final-calc} uses $x+\eps\le1$  and $1-x\le1$. 
\end{proof}